\documentclass[10pt, leqno, oneside]{amsart}
\usepackage{amsfonts}
\usepackage{amsmath, amsthm, amssymb}
\usepackage{comment} 
\usepackage[authoryear,round]{natbib}

\reversemarginpar
\usepackage[dvips,left=1in, right=1in,
top=1.4in, bottom=1.2in, headsep=0.4in,footskip=0.4in]{geometry}

\usepackage{color}

\DeclareMathOperator*\cov{Cov}
\DeclareMathOperator*\var{Var}
\numberwithin{equation}{section}
\theoremstyle{plain}                
\newtheorem{theorem}{Theorem}[section]
\newtheorem{lemma}[theorem]{Lemma}
\newtheorem{proposition}[theorem]{Proposition}
\newtheorem{corollary}[theorem]{Corollary}
\theoremstyle{definition}           
\newtheorem{definition}[theorem]{Definition}
\newtheorem{example}[theorem]{Example}

\newtheorem{assumption}[theorem]{Assumption}
\theoremstyle{remark}
\newtheorem{remark}[theorem]{Remark}

\newcommand{\argmax}{\operatorname{argmax}}

\DeclareMathOperator*\esssup{esssup}

\newcommand{\tot}{\tfrac{1}{2}} 
\newcommand{\abs}[1]{\left| #1 \right|} 
\newcommand{\set}[1]{\left\{#1\right\}} 
\newcommand{\sets}[2]{\set{#1\,:\,#2}} 
\newcommand{\sq}[2]{( #2 )_{#1\in\N}}  
\newcommand{\norm}[1]{{||#1||}} 

\newcommand{\prf}[1]{ ( #1 )_{t\in [0,T]}}

\providecommand{\R}{} \renewcommand{\R}{{\mathbb R}}
 
\newcommand{\N}{{\mathbb N}}
\newcommand{\PP}{{\mathbb P}}
\newcommand{\QQ}{{\mathbb Q}}

\newcommand{\EE}{{\mathbb E}}
\newcommand{\FF}{{\mathcal F}}
\newcommand{\GG}{{\mathcal G}}

\newcommand{\HH}{{\mathcal H}}

\newcommand{\MM}{{\mathcal M}}

\newcommand{\EN}{{\mathcal E}}

\renewcommand{\AA}{{\mathcal A}}

\newcommand{\ba}{{\mathrm{ba}}}

\newcommand{\eps}{\varepsilon}
\newcommand{\ld}{\lambda}

\newcommand{\el}{{\mathbb L}} 
\newcommand{\lzer}{\el^0}
\newcommand{\lone}{\el^1}
\newcommand{\ltwo}{\el^2}

\newcommand{\linf}{\el^{\infty}}

\newcommand{\remove}[1]{\st{#1}}
\renewcommand{\remove}[1]{}
\newcounter{notenum}

\newcommand{\define}[1]{{\em #1}}

\newcommand{\ab}[1]{\langle #1 \rangle}
\newcommand{\bab}[1]{\big\langle #1 \big\rangle}
\newcommand{\Bab}[1]{\Big\langle #1 \Big\rangle}

\newcommand{\bS}{\mathbf{S}}
\newcommand{\bT}{\boldsymbol{\Theta}}
\newcommand{\bt}{\boldsymbol{\vartheta}}
\newcommand{\bvt}{\boldsymbol{\vartheta}}
\newcommand{\bB}{\boldsymbol{B}}

\renewcommand{\ba}{\boldsymbol{\alpha}}

\newcommand{\urrw}[2]{R^{(w)}(#2;#1)}
\newcommand{\urrwg}[1]{\urrw{\gamma}{#1}}

\newcommand{\rrw}[3]{R^{(w)}(#2;#1|#3)}
\newcommand{\rrb}[3]{R^{(b)}(#2;#1|#3)}
\newcommand{\rrwg}[1]{\rrw{\gamma}{#1}{\EN}}

\newcommand{\turrwg}[1]{\urrw{\gamma}{#1}}

\newcommand{\trrw}[3]{R_t^{(w)}(#2;#1|#3)}
\newcommand{\trrwg}[1]{\trrw{\gamma}{#1}{\EN}}

\newcommand{\Trrw}[3]{R_T^{(w)}(#2;#1|#3)}

\newcommand{\ullw}[2]{L^{(w)}(#2;#1)}
\newcommand{\Tullw}[2]{L_T^{(w)}(#2;#1)}

\newcommand{\ullwg}[1]{\ullw{\gamma}{#1}}

\newcommand{\llw}[3]{L^{(w)}(#2;#1|#3)}
\newcommand{\llwg}[1]{\llw{\gamma}{#1}{\EN}}

\newcommand{\tullw}[2]{L_t^{(w)}(#2;#1)}

\newcommand{\tllw}[3]{L_t^{(w)}(#2;#1|#3)}
\newcommand{\tllwg}[1]{\tllw{\gamma}{#1}{\EN}}

\newcommand{\ufun}[3]{u_{#1}(#2|#3)}
\newcommand{\ug}[1]{\ufun{\gamma}{#1}{\EN}}
\newcommand{\uone}[1]{\ufun{\gamma_1}{#1}{\EN_1}}
\newcommand{\utwo}[1]{\ufun{\gamma_2}{#1}{\EN_2}}
\newcommand{\ui}[1]{\ufun{\gamma_i}{#1}{\EN_i}}

\newcommand{\bu}{\boldsymbol{u}}
\newcommand{\bU}{\boldsymbol{U}}

\newcommand{\Pna}{{\mathcal P}^{NA}}
\newcommand{\Pu}{{\mathcal P}^{U}}
\newcommand{\Pui}{{\mathcal P}^{U}_i}

\newcommand{\bpr}{\boldsymbol{p}}
\newcommand{\afun}[2]{\AA_{#1}(#2)}
\newcommand{\ag}{\afun{\gamma}{\EN}}
\newcommand{\aone}{\afun{\gamma_1}{\EN_1}}
\newcommand{\atwo}{\afun{\gamma_2}{\EN_2}}

\newcommand{\cafun}[2]{\AA^{\circ}_{#1}(#2)}
\newcommand{\cag}{\cafun{\gamma}{\EN}}
\newcommand{\caone}{\cafun{\gamma_1}{\EN_1}}
\newcommand{\catwo}{\cafun{\gamma_2}{\EN_2}}

\newcommand{\nuw}[3]{\nu^{(w)}(#2;#1|#3)}

\newcommand{\tnuw}[3]{\nu_t^{(w)}(#2;#1|#3)}
\newcommand{\tnuwg}[1]{\tnuw{\gamma}{#1}{\EN}}

\newcommand{\tnub}[3]{\nu_t^{(b)}(#2;#1|#3)}
\newcommand{\tnubg}[1]{\tnub{\gamma}{#1}{\EN}}

\newcommand{\nub}[3]{\nu^{(b)}(#2;#1|#3)}
\newcommand{\nuwone}[1]{\nuw{\gamma_1}{#1}{\EN_1}}
\newcommand{\nubtwo}[1]{\nub{\gamma_2}{#1}{\EN_2}}
\newcommand{\nubone}[1]{\nub{\gamma_1}{#1}{\EN_1}}
\newcommand{\nuwi}[1]{\nuw{\gamma_i}{#1}{\EN_i}}
\newcommand{\nubi}[1]{\nub{\gamma_i}{#1}{\EN_i}}

\newcommand{\nuwg}[1]{\nuw{\gamma}{#1}{\EN}}
\newcommand{\nubg}[1]{\nub{\gamma}{#1}{\EN}}

\newcommand{\tgamma}{\tilde{\gamma}}

\newcommand{\unuw}[2]{\nu^{\left(w\right)}\left(#2;#1\right)}
\newcommand{\unub}[2]{\nu^{\left(b\right)}\left(#2;#1\right)}

\newcommand{\unuwone}[1]{\unuw{\gamma_1}{#1}}
\newcommand{\unuwtwo}[1]{\unuw{\gamma_2}{#1}}
\newcommand{\unubone}[1]{\unub{\gamma_1}{#1}}
\newcommand{\unubtwo}[1]{\unub{\gamma_2}{#1}}

\newcommand{\unuwg}[1]{\unuw{\gamma}{#1}}
\newcommand{\unubg}[1]{\unub{\gamma}{#1}}
\newcommand{\unuwb}[1]{\unuw{\bar{\gamma}}{#1}}

\newcommand{\rinf}{{\mathcal R}^{\infty}}

\newcommand{\agset}{\mathcal{G}_{\mathcal{E}_{1},\mathcal{E}_{2}}}
\newcommand{\cagset}{\mathcal{G}^{\circ}_{\mathcal{E}_{1},\mathcal{E}_{2}}}
\renewcommand{\agset}{\mathcal{G}}
\renewcommand{\cagset}{\mathcal{G}^{\circ}}

\newcommand{\QB}{\QQ^{(B)}}
\newcommand{\QZ}{\QQ^{(0)}}

\newcommand{\hq}[1]{\QQ^{(#1)}}
\newcommand{\hqi}[1]{\QQ_i^{(#1)}}
\newcommand{\hqone}[1]{\QQ_1^{(#1)}}
\newcommand{\hqtwo}[1]{\QQ_2^{(#1)}}

\DeclareMathOperator{\Var}{Var}

\newcommand{\bD}{\boldsymbol{\Delta}}
\newcommand{\bdelta}{\boldsymbol{\delta}}

\newcommand{\eR}{\overline{\R}}
\newcommand{\bz}{\boldsymbol{0}}

\begin{document}

\begin{center}
  \LARGE\bf On Agents' Agreement and Partial-Equilibrium Pricing in
  Incomplete Markets\footnote{Both authors were supported in part by
    the National Science Foundation under award number DMS-0706947
    during the preparation of this work. Any opinions, findings and
    conclusions or recommendations expressed in this material are
    those of the authors and do not necessarily reflect those of the
    National Science Foundation.

    The authors would like to thank Hans F{\"o}llmer, Lorenzo
    Garlappi, Stathis Tompaidis, Thaleia Zariphopoulou, and the
    participants of the 11th International Congress on Insurance:
    Mathematics and Economics, IME, Piraeus, Greece, July 2007 for
    fruitful discussions and good advice.}
\end{center}
\ \\[-1.5ex]
\begin{center} \today \end{center}
\ \\[0.5ex]

\begin{center}
\begin{minipage}{0.4\textwidth}
\begin{center}
{\bf\large Michail Anthropelos}\\
Department of Mathematics\\
University of Texas at Austin\\
1 University Station, C1200\\
Austin, TX 78712, USA\\
{\tt manthropelos@math.utexas.edu}\\
\end{center}
\end{minipage}
\begin{minipage}{0.4\textwidth}
\begin{center}
{\bf\large Gordan \v Zitkovi\' c}\\
Department of Mathematics\\
University of Texas at Austin\\
1 University Station, C1200\\
Austin, TX  78712, USA\\
{\tt gordanz@math.utexas.edu}\\
\end{center}
\end{minipage}
\end{center}

\ \\[1ex]

\begin{quote}
\noindent{\bf Abstract.} We consider two risk-averse financial agents who negotiate the price
  of an illiquid indivisible contingent claim in an incomplete
  semimartingale market environment. Under the assumption that the
  agents are exponential utility maximizers with non-traded random
  endowments, we provide necessary and sufficient conditions for
  negotiation to be successful, i.e., for the trade to occur.  We also
  study the asymptotic case where the size of the claim is small
  compared to the random endowments and we give a full
  characterization in this case. Finally, we study a partial-equilibrium
  problem for a bundle of divisible claims and establish existence and
  uniqueness. A number of technical results on conditional
  indifference prices is provided.
 \end{quote}

 \ \\[-0.1ex]


\ \\[-6ex]


\ \\[3ex]

\section{Introduction}

\subsection{A description of the problem}
In an ideal complete market, each contingent claim can be perfectly
replicated and, thus, a rational agent is indifferent between the
(random) claim itself and its (deterministic) replication price.
Abundant empirical evidence shows that the real financial markets are
far from complete; only a small portion of contingent claims can be
replicated in the market to a satisfactory degree. A non-specific abstract
notion of rationality is no longer sufficient to single out a unique
``fair price'' of
any contingent claim. This effect is demonstrable  in over-the-counter
transactions where (typically) two agents negotiate a price of a
single, indivisible, not-perfectly-replicable contingent claim. The
final outcome of such a negotiation eventually hinges upon two
idiosyncratic factors - the agents' attitude towards risk and their
negotiation skills. The focus of the present paper is the former. We
ask the following question: {\em Under what conditions on the claim
whose price is being negotiated, the liquid-market
environment and the agents' risk attitudes will
a mutually beneficial agreement be feasible?}

Our modelling choices are informed by simplicity, but we steer clear of
oversimplification.  In particular, we assume that the two agents are
expected utility maximizers in the von Neumann-Morgenstern sense, with
a common investment horizon $T$.  For simplicity and analytic
tractability we assume that both agents' utility functions are
exponential, possibly with different risk-aversion parameters.  An
important feature which is not present in a major part of the past
work on the subject is the presence of random endowments - the agents
are assumed to hold an illiquid portfolio and the risk assessment of
any contingent claim will depend heavily upon its (co-)relation with
this illiquid portfolio.  In addition to the illiquid random
endowments, both agents have access to a liquid incomplete financial
market modelled by a general locally-bounded semimartingale. Also, we
assume that all pay-offs are already discounted in time-$0$ terms;
this way we can freely compare values corresponding to different
points in time. Mathematical-finance literature abounds with
information on the utility-maximization problem for a variety of
utility concepts (see, for instance, \citet{KarLehShrXu90},
\citet{KraSch99}, \citet{Sch01}, \citet{CviSchWan01}, \citet{Owe02},
\citet{OweZit06}).

Under the conditions described above, the two agents meet at time
$0$ when one of the agents (the seller)
offers a contingent claim with time-$T$
payoff $B$ to the other one (the buyer)
in exchange for a lump-sum payment $p$ at
time $t=0$. Our central question, posed above, can now be made more
precise and split into two separate components:
\begin{itemize}
\item[1.] Is there a number $p\in\R$ such that the exchange of
  the contingent claim $B$ for a lump sum $p$ is (strictly) beneficial for
 both agents?
\item[2.] If more than one such $p$ exists, can we determine the exact
  outcome of the negotiation?
\end{itemize}
The net gain $B-p$ will be beneficial for the buyer if he/she can find a
trading strategy such that the resulting wealth at time $T$ gives rise
to a higher expected utility than the one he/she would be able to
obtain without $B-p$. A similar criterion applies to the seller.
In case the answer to question 1.~is positive, we say that the agents
are {\em in agreement}.

While we give a fairly complete answer to  question 1., we only
touch upon the issues involved in question 2. In fact, it is not possible to
give a definitive  answer to this question without a precise model of the
negotiation process (see, for instance, \citet{BazNea92}). A partial answer is
possible,
however, when
the indivisibility assumption is dropped (see Section
\ref{sec:pepq}).

\subsection{Our results and how they relate to existing research} Our
results are naturally split into 4 parts which correspond to Sections
3, 4, 5 and Appendix A in the present paper:

\medskip

\noindent{\em 1.~Abstract agreement: } We start with the study of the class
$\cagset$ of all (appropriately regular) contingent claims $B$ for which the
  agents are in agreement. It is, perhaps, surprising that unless
  non-replicable random endowments are present, no contingent claims
  will lead to agreement, even for agents with  different
  risk-aversion coefficients.  When the random endowments are indeed
  present, we give a necessary and sufficient condition for the set
  $\cagset$ to be non-empty.
  This characterization is closely related
  to the notion of {\em optimal risk sharing} which was first studied
  in the context of  insurance/reinsurance negotiation (see, for instance,
  \citet{BuhJew79}, \citet{DanSca07}) and recently developed for in more general
  settings  (see \citet{FolSch04}, \citet{BarElk04},
  \citet{JouSchTou06} and \citet{FilKup08}).

\medskip

\noindent{\em 2.~Agreement for a specific claim - residual risk and
 approximation: }
Next, we consider a question which is in a sense dual to the one
tackled in the previous part: is there a criterion for an agreement
about a {\em given} claim $B$?  We propose two approaches: one through
the notion of residual risk and the other based on  asymptotic
approximation of conditional indifference prices for small quantities.

 Residual risk
(introduced in \citet{MusZar04a}) of a
random liability is defined as the difference between the
liability's payoff and the terminal
value of the optimal risk-monitoring strategy
at  maturity. We establish the following criterion, made precise in
the body of the paper:
a claim is mutually agreeable if and only if it reduces the  residual
 risk for  both agents.

 The other approach provides an explicit criterion in the asymptotic
 case when the size of the contingent claim $B$ is small compared to
 the size of the agents' random endowments. It is possible to
 phrase the agreement problem in terms of a relationship between the
 buyer's and the seller's conditional indifference price for the
 claim, so it is not unusual that an asymptotic study of these
 quantities plays a major role. More precisely, we establish a rather
 general Taylor-type approximation of the conditional exponential
 indifference price for locally bounded semimartingales on
 left-continuous filtrations. These approximations are  then used to
 give simple asymptotic criteria for agreeability, as well as
 the asymptotic size
 of the interval of mutually-agreeable prices.
Since it is not possible to obtain closed-form
representations of indifference prices in general market models, such
asymptotic results can be very useful even beyond the agreement
problem.

Asymptotic techniques  are not new in utility maximization problems. In
\citet{KraSir07}, a first order approximation of the optimal
hedging strategy in a semimartingale market for general utilities
(defined on the positive real line) is provided. This generalizes
the results of \citet{Hen02} and \citet{HenHob02}. For
exponential utility, the first derivative of the indifference
price for a vector of claims is given in \citet{IlhJonSir05}. By
imposing the assumption of left-continuity on the filtration,
 we generalize
their result (as well as the asymptotic approximation in  \citet{SirZar05})
by providing a second order approximation of the
price for a vector of claims.

\medskip

\noindent{\em 3.~Partial equilibrium prices: }
In the third part of the paper we look into the following, related,
question: {\em if the agents are allowed to choose not only the price of
the claim, but also the quantity traded, can the market clearing
(partial equilibrium)
conditions be used to compute these two quantities?}
We consider bundles of several contingent claims and prove
existence and uniqueness of equilibrium price-quantities, as
well as a formula for the partial equilibrium price.
The existence results of various types of competitive equilibria are a staple
of quantitative economics literature, and have recently made their way
into mathematical finance (see, among others,
\citet{DanLeV00}, \citet{HeaKu04}, \citet{Zit06},
\citet{BurRus07} and  \citet{FilKup08a}). Our incomplete
partial-equilibrium setting is, however, new
and not covered by any of the existing results.
As we already mentioned above, it is only in the present setting that we
can say something about question 2., i.e., about
the realized price $p$ of the offered contingent
claim $B$.

\medskip

\noindent{\em 4.~Conditional indifference prices:}
Our structural results rely heavily on the notion of conditional
indifference prices. Utility-indifference prices were first introduced
in \citet{HodNeu89}, and then further investigated
and developed by a large number of authors (see, for instance,
\citet{KarKou96}, \citet{Dav97}, \citet{Fri00}, \citet{RouElk00},
\citet{Zar01a}, \citet{HugKraSch05}, \citet{ManSch05},
\citet{KloSch07}).  The special case (pertinent to the present paper)
of exponential indifference prices was studied, e.g.,  in
\citet{Fri00}, \citet{RouElk00}, \citet{Zar01a},
\citet{DelGraRheSamSchStr02} and  \citet{ManSch05}.

In the presence of an illiquid random endowment, we talk about a
conditional indifference price (also known as relative indifference
price in \citet{MusZar03} and \citet{Sto06}). In the exponential world,
some of its properties can be obtained by a simple change of measure
which, effectively, removes the conditionality. Other properties,
however, cannot be dealt with in that manner. The goal of the last
part of this work is to establish certain properties
 of conditional indifference prices in
a general semimartingale market setting. We show, for instance, a
rather
unexpected fact that conditional indifference prices (unlike their
unconditional versions) do not have
to be monotone in the risk-aversion parameter.

\subsection{The structure of the paper} In Section \ref{sec:model}, we
describe the market model and introduce necessary notation. The notion
of agreement is introduced and our main abstract results are proven in
Section \ref{sec:agree}. In Section \ref{sec:approximation} we study
the small-quantity asymptotics of conditional indifference prices, and
use it to provide an agreement criterion. An example in the Brownian
setting is also presented. The topics  of Section \ref{sec:pepq} are
existence and uniqueness of partial equilibrium price-quantities for
vectors of contingent claims.  In Appendix \ref{sec:cond-price} we
state some properties of conditional indifference prices, and in
Appendix \ref{sec:res-risk} we give an outline of some known
results on residual risk.

\bigskip

\section{Some modelling and notational preliminaries}
\label{sec:model}
\subsection{The financial market}

Our model of the financial market is based on a filtered probability
space $(\Omega ,\mathcal{F},\mathbb{F },\PP)$, $\mathbb{F}=\left(
  \mathcal{F}_{t}\right) _{t\in [0,T]}$, $T>0$, which satisfies the
usual conditions of right-continuity and completeness. There are
$d+1$ traded assets ($d\in\mathbb{N}$), whose discounted price
processes are modelled by an $\mathbb{R}^{d+1}$-valued locally
bounded
semimartingale 
$(S^{(0)}_t;\bS_t)_{t\in [0,T]} =(S^{(0)}_{t} ;
S^{(1)}_{t},\dots,S^{(d)}_{t}) _{t\in[ 0,T ] }$. The first asset
$S^{(0)}_{t}$ plays the role of a num\'{e}raire security or a
discount factor.  Operationally, we simply set $S^{(0)}_t\equiv
1$, for all $t\in [0,T]$, a.s.

\subsection{Agent behaviour}
Placing ourselves in the von Neumann-Morgenstern framework, we
assume that each market participant evaluates the risk of an
uncertain position $X$ at time $T$ according to the expected
utility $\EE^{\PP}[U(X+\EN)]$, where $U$ is a utility function and
$\EN$ is the \define{random endowment} (accumulated illiquid wealth)
and $\PP$ is a subjective probability measure.
For technical reasons, we restrict our attention to
$\EN\in\linf(\FF)$ and the class of exponential utilities
\[ U(x)=-\exp(-\gamma x),\ x\in \mathbb{R}\] where the constant
$\gamma\in (0,\infty)$ 
is the (absolute) risk aversion coefficient.

\subsection{Admissible strategies and the absence of arbitrage}

A financial agent invests in the market by choosing a portfolio
strategy $\bvt$ in  an admissibility class $\bT$, to be specified
below.  The resulting \define{gains process} $\prf{G^{\bt}_t}$ is
simply the stochastic integral $G^{\bt}_t=(\bvt\cdot
\bS)_t=\int_0^t \bvt_u\, d\bS_u$.  Due to the exponential nature
of the utility functions considered here, we follow the setup
introduced in \citet{ManSch05} or \citet{DelGraRheSamSchStr02}.
Before we give a precise description of the aforementioned set
$\bT$, we need to introduce several concepts related to the
no-arbitrage requirement. We start with the set $\MM_a$
of \define{absolutely continuous local martingale measures}, where
\[\MM_a=\sets{\QQ\ll\PP}{ \bS\text{ is a local martingale under }
  \QQ}\]The set $\MM_e$
 of all elements $\QQ$ of
$\MM_a$ which additionally satisfy $\QQ\sim\PP$ is called \define{the
  set of equivalent local martingale measures}. For a probability
measure $\QQ$ on $(\Omega ,\mathcal{F})$, we define
\[\mathcal{H}(\QQ|\PP)
=\begin{cases}
\EE^{\PP}\left[ \frac{d\QQ}{d\PP}\ln \left( \frac{d\QQ}{d\PP}\right) \right] &
\QQ\ll\PP,\\
+\infty, & \text{otherwise}
\end{cases}
\]
The (extended) positive number $\HH(\QQ|\PP)$ is called the
\define{relative entropy} of the probability measure $\QQ$ with respect to probability measure $\PP$.
For details on the notion of relative entropy we refer the
interested reader to \citet{GraRhe02} or \citet{Fri00}.  We set
\[\MM_{e,f}
=\sets{\QQ\in\MM_e}{\HH(\QQ|\PP)<\infty}\] and enforce the
following assumption
\begin{assumption}
\label{ass:NA}
$\MM_{e,f}\not=\emptyset$.
\end{assumption}
Assumption \ref{ass:NA} trivially implies that
$\MM_e\not=\emptyset$ which, in turn, guarantees that no arbitrage
opportunities exist in the market (a stronger statement of ``no
free lunch with vanishing risk'' will hold, as well). The additional requirement in Assumption \ref{ass:NA}
is common in the literature and it
ensures that the choice of the exponential function for the
utility leads to a well-defined behavior for utility-maximizing
agents (see, among others, \citet{DelGraRheSamSchStr02},
\citet{Fri00}, \citet{Bec01a} and \citet{ManSch05}).

Having introduced the required families of probability measures, we
turn back to definition of the class $\bT$ of \define{admissible strategies}:
\begin{equation}\label{strategies}
\bT =\left\{ \mathbf{\bvt }\in L\left(
\bS\right) :\left(
\mathbf{\bvt }\cdot \bS\right) \text{ is a }\QQ \text{-martingale, }%
\forall\,\QQ\in \mathcal{M}_{e,f}\right\}
\end{equation}
where $L(\bS)$ is the set of all predictable $(d+1)$-dimensional
$\bS$-integrable processes on $[0,T]$. More information about the
set $\bT$ of admissible strategies is given in \citet{ManSch05}
(see also remarks on the set $\bT_2$ in
\citet{DelGraRheSamSchStr02}).

We remind the reader that $\lzer(\FF)$ denotes the set of all
($\PP$-a.s. equivalent classes of) $\FF$-measurable random
variables. A random variable $B\in \lzer(\FF)$ is said to be
\define{replicable} if there exists a constant $c$ and an
admissible strategy $\mathbf{\bvt }\in \bT $ such that $B=c+\left(
\bvt \cdot \bS\right) _{T}$ a.s.; the set of replicable random
variables will be denoted by $\mathcal{R}$.
More generally, we introduce the following equivalence relation
between random variables in $\lzer(\FF)$:
\begin{definition}\label{equiv.class}
We call two random variables $B,C\in \lzer(\FF)$ \emph{risk
equivalent} or \define{equal up to replicability} and write $B\sim C$,
if the difference $B-C$ is
replicable.
\end{definition}

It is clear that the relation $\sim $ is an equivalence relation
on $\lzer(\FF)$ (since $\bT $ is a vector space).
We note that the zero equivalence class coincides with the set
$\mathcal{R}$ of the replicable random variables. For future
reference, we let $\rinf={\mathcal R}\cap \linf(\FF)$ denote the
set of all (essentially) bounded replicable random variables.

\subsection{Some special probability measures}
\label{sse:special-measures} The expectation operator under a
probability measure $\QQ$ is denoted by $\EE^{\QQ}[\cdot] $, where
the superscript $\QQ$ is omitted in the case of the (subjective)
measure $\PP$. Also, for a random vector
$\mathbf{B}=(B_{1},B_{2},...,B_{n})$,
$\EE^{\QQ}[\mathbf{B}]$ stands for the vector
$(\EE^{\QQ}[B_{1}],\EE^{\QQ}[B_{2}],...,\EE^{\QQ}[B_{n}])\in\R^n$.

For a random variable $B\in \lzer(\FF)$ with
$\EE[\exp(B)]<\infty$, the probability measure whose Radon-Nikodym
derivative with respect to $\PP$ is given by $\tfrac{\exp(B)}{\EE[
  \exp(B)]} $, is denoted by $\PP_B$.
Furthermore, $\QZ$ denotes the probability measure in $\MM_a$ with
the minimal relative entropy with respect to $\PP$ i.e., the
probability measure for which $\HH(\QZ|\PP)\leq \HH(\QQ|\PP)$ for
all $\QQ\in\MM_a$. It is a consequence of Assumption \ref{ass:NA}
that the probability measure $\QZ$ exists,  is unique and belongs to
$\mathcal{M}_{e,f}$ (see \citet{Fri00}, page 43, Theorem 2.2).
Similarly, for every $B$ such that $\EE[\exp(B)]<\infty$, there
exists a unique probability measure $\QB \label{mentropymeasureB}
\in\MM_a$ such that $\HH(\QB|\PP_B)\leq \HH(\QQ|\PP_B)$ for all
$\QQ\in\MM_a$ (see \citet{DelGraRheSamSchStr02}, page 103).

\bigskip

\section{A notion of agreement between financial
agents}
\label{sec:agree}
\subsection{Utility maximization and indirect utility}  Given their
risk profiles, financial agents  trade in the financial market
with the goal of maximizing expected utility. More precisely, an
agent with initial wealth $x\in\R$, risk-aversion coefficient $\gamma$ and random
endowment $\EN\in\linf$   will choose a portfolio process
$\bvt\in\bT$ so as to maximize the expected utility $\EE[
-\exp(-\gamma( (x+\bvt\cdot\bS)_T+\EN))]$.
The value function $\ug{x}$ of the corresponding optimization
problem is given by
\begin{equation}
   \label{equ:value-function}
   \begin{split}
  \ug{x}=\underset{\bvt \in \bT }{\sup} \EE\Big[
  -\exp\left(-\gamma\big( x+( \bvt \cdot S)_{T}+\EN \big) \right)
  \Big],\ x\in\R.
   \end{split}
\end{equation}
Overloading the notation slightly, for any random variable $B\in
\linf(\FF)$ (interpreted as a contingent payoff with maturity $T$)
we define the
\define{indirect utility} of
  $B$   by $\ug{B}=\ufun{\gamma}{0}{\EN+B}$, i.e.
\begin{equation}\label{equ:value-function-B}
\ug{B}=\underset{\bvt \in \bT }{\sup }\EE \Big[ -\exp(-\gamma
\left( \left( \bvt \cdot S\right)_{T}+\EN +B\right)) \Big].
\end{equation}
\begin{remark}
Thanks to the choice of the exponential utility, the case where
the agents have different subjective probability measures, say
$\PP_1\neq\PP_2$, is also covered. Indeed, if we assume that
$\PP_1\approx\PP_2$ and $\ln(\frac{d\PP_1}{d\PP_2})\in\linf$, we
can reduce the analysis to the case of two
 agents with the same subjective measure, say $\PP_2$,
by adding $\gamma_1\ln(\frac{d\PP_1}{d\PP_2})$ to first
agent's random endowment.
\end{remark}
\subsection{A preference relation and a notion of acceptability}
The indirect utility $\ug{\cdot}$ induces a preference
relation $\preceq_{\gamma,\EN}$, on  $\linf(\FF)$;
for $B_1,B_2\in\linf(\FF)$, we set
\[ B_1 \preceq_{\gamma,\EN} B_2\ \text{ if }\  \ug{B_1} \leq \ug{B_2}.\]
In words, the payoff $B_2$ is  preferable
to the payoff $B_1$ for the agent with random endowment $\EN$ and
risk aversion coefficient $\gamma$, if the total payoff $\EN +B_2$
yields more indirect utility than the payoff $\EN+B_1$.

The set of all payoffs $B\in\linf(\FF)$ such that $0
\preceq_{\gamma,\EN} B$ is called the
$(\gamma,\EN)$-\define{acceptance set}, and we denote it by $\ag$.
Equivalently, we have
\begin{equation}
\label{equ:AA}
   \begin{split}
  \ag=\Big\{B\in \linf(\FF)\,:\,&
\underset{\bvt \in \bT }{\sup }\EE
\Big[
-\exp(-\gamma \left( \left( \bvt \cdot S\right)_{T}+\EN\right))
\Big]\leq\\\leq &\underset{\bvt\in \bT }{\sup }\EE
\Big[
-\exp(-\gamma \left( \left( \bvt \cdot S\right)_{T}+\EN+B\right))
\Big]
\Big\} .
   \end{split}
\end{equation}
Closely related to the relation $\preceq_{\gamma,\EN}$ is its strict version
$\prec_{\gamma,\EN}$ defined by
\[ B_1 \prec_{\gamma,\EN} B_2\ \text{ if }\  \ug{B_1} < \ug{B_2},\
B_1,B_2\in \linf(\FF).\]
The \define{strict $(\gamma,\EN)$-acceptance set} $\cag$, is
 defined
by $\cag=\sets{B\in \linf}{\ug{0}<
  \ug{B}}$, and the representation analogous to (\ref{equ:AA}) (with
$\leq$ replaced by $<$) holds.

It is a consequence of the choice of the exponential utility
function that the addition of any constant initial wealth $x\in\R$
to random endowment $\EN$ does not influence the acceptance sets
$\ag$ and $\cag$. More generally, we have the following simple
proposition, the proof of which is standard.
We remind the reader that a set $\mathcal{A}$
is called \emph{monotone} if $B\geq C$, a.s.~and $C\in\mathcal{A}$ imply
$B\in\mathcal{A}$.

\begin{proposition}\label{pro:agree-set}
For every $\EN\in\linf$ and $\gamma\in (0,\infty)$, the sets $\ag$
and $\cag$ are convex, monotone and
\begin{equation}\label{equ:AA_1=AA_2}
\EN_1\sim\EN_2 \text{ implies }
\afun{\gamma}{\EN_1}=\afun{\gamma}{\EN_2} \text{ and
}\cafun{\gamma}{\EN_1}=\cafun{\gamma}{\EN_2}
\end{equation}
\end{proposition}

\subsection{Conditional indifference prices}
The acceptance set $\ag$ can be used to introduce the notion of a
conditional indifference price. The
\define{conditional writer's indifference price} $\nuwg{B}$
of the contingent claim $B\in
\linf(\FF)$ is defined by
\begin{equation}
   \label{equ:nuw-def}
   \begin{split}
 \nuwg{B}=\inf \sets{p\in\R}{p-B\in \ag}.
   \end{split}
\end{equation}
i.e., $\nuwg{B}$ is the minimum amount that the agent with
preference relation $\preceq_{\gamma,\EN}$ will be willing to sell
the claim with payoff $B$ for.
Similarly, the
\define{conditional buyer's indifference price} $\nubg{B}$ is
defined by
\begin{equation}
   \label{equ:nub-def}
   \begin{split}
 \nubg{B}=\sup \sets{p\in\R}{B-p\in \ag}.
   \end{split}
\end{equation}
i.e., $\nubg{B}$ is the maximum amount that the agent with
preference relation $\preceq_{\gamma,\EN}$ will offer for a
contingent claim with payoff $B$.

In the special case where $\EN\sim 0$, the corresponding prices
are called \textit{unconditional indifference prices} (or, simply,
indifference prices) and are denoted by $\unuwg{B}$ and
$\unubg{B}$. A compendium of relevant properties of both conditional and unconditional indifference
prices is given in Appendix A.

The notion of the indifference price has been studied by many
authors (see, among others, \citet{HodNeu89},
\citet{RouElk00},\citet{Hen02} and \citet{MusZar04a}). The definition
of the conditional indifference price under exponential utility
was given in \citet{Bec01a} for general semimartingale model, in
\citet{MusZar05} for a binomial case model and in \citet{Sto06} for
a diffusion model (where the price is called relative indifference
price). A discussion of the conditional indifference price under
general utility functions is given in \citet{OweZit06}.

\subsection{Agreement}

The present paper deals with the interaction between two
  financial agents,  with risk aversion coefficients
$\gamma_1$ and $\gamma_2$  and random
endowments $\EN_1,\EN_2\in \linf(\FF)$.


\begin{definition}
\label{def:agreement}
A contingent claim $B\in \linf(\FF)$ is said to be
\begin{enumerate}
\item \define{mutually agreeable} if there exists a number
$p\in\R$ such that $p-B\in \aone$ and $B-p\in\atwo$. \item
\define{strictly mutually agreeable} if there exists a number
$p\in\R$ such that $p-B\in \caone$ and $B-p\in\catwo$.
\end{enumerate}
If a claim $B$ is (strictly) mutually agreeable, the set of all
$p\in \R$ such that the conditions in (1) (or (2) in the strict
case) above hold is called the \define{set of (strictly) mutually
agreeable prices} for $B$.
\end{definition}

A discussion related to our notion of
 mutually agreeability
is given in \citet{JouSchTou06}, subsection 3.6, for cash
invariant monetary utility functions, but without the presence of
a financial market.

Using the conditional writer's and buyer's indifference prices,
$\nuwone{\cdot}$ and $\nubtwo{\cdot}$
defined above, we can give a simple characterization of the set of
mutually-agreeable prices.
\begin{proposition}
\label{pro:agr-prices}
A claim $B\in\linf(\FF)$ is mutually
agreeable if and only if
\begin{equation} \label{agreementcondition}
\nuwone{B}\leq \nubtwo{B}.\end{equation}
In that case, the set of
mutually-agreeable prices for $B$ is given by
\[[\nuwone{B},\nubtwo{B}].\]
\end{proposition}
\begin{remark}\
\label{rem:agr-prices}
\begin{enumerate}
\item A version of Proposition \ref{pro:agr-prices} for strict
  mutually-agreeable prices with strict inequality in
  (\ref{agreementcondition}) and the interval
  $[\nuwone{B},\nubtwo{B}]$ replaced by
  its interior $(\nuwone{B},\nubtwo{B})$
  holds.
\item \label{ite:ag-prices-in-NA} For a contingent claim
$B\in\linf(\FF)\setminus\rinf$,
 each (strictly) mutually agreeable price $p$ of
  $B$ satisfies $p\in (\inf_{\QQ\in\MM_a} \EE^{\QQ}[B],
\sup_{\QQ\in\MM_a} \EE^{\QQ}[B])$ (see, e.g., \citet{OweZit06},
Proposition 7.2) i.e., every mutually agreeable price is an
arbitrage-free price. Trivially, every claim $B\in\rinf$ is
mutually agreeable and the mutually agreeable price is unique and
equal to the unique arbitrage-free price.
\end{enumerate}
\end{remark}

\subsection{The set of all mutually agreeable claims}
It will be important in the sequel to introduce separate notation
for the set of all (strictly) mutually agreeable claims:
\begin{equation}
   \label{equ:ag.set}
   \nonumber
   \begin{split}
     \agset & =\sets{B\in\linf}{ B \text{ is mutually agreeable}},
     \text{ and}, \\
     \cagset & =\sets{B\in\linf}{ B \text{ is strictly mutually
         agreeable}}.
   \end{split}
\end{equation}
 We
remind the reader that $\rinf$ is the set of all replicable claims
in $\linf(\FF)$.
\begin{proposition}\ \label{pro:first-prop-G}
\begin{enumerate}
\item $\agset$ is convex and $\sigma(\linf,\lone)$-closed.
\item $\agset\cap (-\agset)=\rinf$, $\cagset\cap
  (-\cagset)=\emptyset$,
\item $\agset=\linf$ if and only if $\rinf=\linf$.
\end{enumerate}
\end{proposition}
\begin{proof}\
\begin{enumerate}
\item The convexity of $\agset$ follows from the convexity of
  $\nuwone{\cdot}$ and concavity of $\nubtwo{\cdot}$
(see Proposition \ref{pro:risk-measures}). As for the closedness,
it will be enough to note that $\nuwone{\cdot}:\linf\to\R$ is
lower semi-continuous and $\nubtwo{\cdot}:\linf\to\R$ is upper
semi-continuous with respect to the weak-* topology
$\sigma(\linf,\lone)$ (see Corollary \ref{cor:nu-sc}). \item
Trivially, $\rinf\subseteq\agset\cap (-\agset)$. For a claim $B\in
\agset \cap (-\agset)$, there exists $p,\hat{p}\in \R$ such that
$p-B \in \aone$  and $B-p\in \atwo$ , as well as $\hat{p}+B \in
\aone$ and $-B-\hat{p}\in \atwo$. It follows, by convexity of
$\aone$ that $ \tot (p+\hat{p})=\frac{1}{2} ( p-B+\hat{p}+B )\in
\aone$, i.e., $\uone{0} \leq \uone{\tot (p+\hat{p})}$. The strict
monotonicity of the value-function $\uone{.}$ for deterministic
arguments implies that $\tot(p+\hat{p})\geq 0$. Applying the same
line of reasoning to $\atwo$ and the value function $\utwo{.}$, we
get that $-\tot(p+\hat{p}) \geq 0$, and, consequently,
$p=-\hat{p}$. Using the definitions (\ref{equ:nuw-def}) and
(\ref{equ:nub-def}) of the conditional indifference prices we
easily get that $ \nubone{B}\geq p \geq \nuwone{B}$, which,
according to Corollary \ref{cor:non-agreement}, implies that
$B\in\rinf$.

To prove the second claim, it suffices to note that
$\cagset\cap\rinf=\emptyset$. Indeed,
$\nuwone{B}=\nubtwo{B}=\EE^{\QQ}[B]$ for $B\in \rinf$ and all
$\QQ\in\MM_a$. \item If $\agset=\linf$ then $\linf\subseteq
\agset\cap (-\agset)$ so $\linf=\rinf$, by (2) above. Conversely,
if $\linf=\rinf$ then $\linf= \agset\cap (-\agset)\subseteq
\agset$.
\end{enumerate}
\end{proof}
\begin{remark}
The weak-* topology $\sigma(\linf, \lone)$ in Proposition
\ref{pro:first-prop-G} can be replaced by an even weaker one,
namely the coarsest topology $\tau$ on $\linf$ which makes the
expectation mappings $\EE^{\QQ}[\cdot]:\linf\to\R$ continuous for
each $\QQ\in\MM_{e,f}$.
\end{remark}

\subsection{No agreement without random endowments}
The following, at first glance surprising, result states that mere
difference in risk-aversion is not enough for two exponential agents
to agree on a price of {\em any} contingent claim. Qualitatively
different random endowments are needed.
\begin{proposition}
[Non-agreement with replicable random endowments] Suppose that
\label{pro:non-agreement}
$\EN_1\sim\EN_2\sim 0$. Then $\agset=\rinf$ and
$\cagset=\emptyset$.
\end{proposition}

\begin{proof} The limiting relationships in (\ref{equ:uncon-asym})
and the monotonicity properties of the indifference prices (see
Proposition \ref{pro:monotone}) imply that
\begin{equation}\label{equ:non-agreement}
\nuwone{B}=\unuwone{B} \geq \EE^{\QZ}[B] \geq
\unubtwo{B}=\nubtwo{B},\forall B\in\linf.
\end{equation}
Therefore, the strict inequality $\nuwone{B} < \nubtwo{B}$ -
needed for the strong agreement - cannot hold. Consequently,
$\cagset=\emptyset$.

If $B\in\agset$, (\ref{equ:non-agreement}) implies that
$\nuwone{B}=\unuwone{B}=\EE^{\QZ}[B]=\lim_{\gamma\to 0}
\unuw{\gamma}{B}$. Therefore, the function $\gamma\mapsto
\unuwg{B}$ can not be strictly increasing on $(0,\infty)$, so, by
Proposition \ref{pro:monotone}, we must have $B\in\rinf$. Hence,
$\agset\subseteq \rinf$, and part (2) of Proposition
\ref{pro:first-prop-G} implies that $\agset=\rinf$.
\end{proof}

Our following result, Corollary \ref{cor:non-agreement}, follows
directly from Proposition \ref{pro:non-agreement} and the fact
that the conditional indifference price becomes unconditional if
the measure $\PP$ is changed to $\PP_{-\gamma\EN}$ (see
Proposition \ref{pro:first-cond} and the discussion at the
beginning of subsection \ref{subsec:Conditional}).
\begin{corollary}\label{cor:non-agreement}
Suppose that $\EN_1\sim\EN_2$. Then for every $\gamma>0$ and
$B\in\linf$, we have
\[\begin{cases}
\nuw{\gamma}{B}{\EN_1}>\nub{\gamma}{B}{\EN_2} &
\text{for }B\notin\rinf,\\
\nuw{\gamma}{B}{\EN_1}=\nub{\gamma}{B}{\EN_2} & \text{otherwise.}
\end{cases}
\]
\end{corollary}

\subsection{Agreement with random endowments}
Proposition \ref{pro:non-agreement} states that the absence of random
endowments is a sufficient condition for the lack of (strict)
agreement. Is it also necessary?
Given the result of Proposition \ref{pro:agr-prices},
the question of the existence of non-replicable mutually agreeable claims
leads to the following optimization problem with
 value function $\Sigma : (0,\infty)^2\times (\linf)^2
\to [0,+\infty]$, where
\begin{equation}
\label{equ:min-diff}
\Sigma(\gamma_1,\gamma_2,\EN_1,\EN_2)
=\sup_{B\in\linf} \big(  \nubtwo{B}- \nuwone{B} \big).
\end{equation}
It follows directly from the Definition \ref{def:agreement} of the
set $\GG$ that the following result holds:
\begin{proposition}
\label{pro:char-agree}
For $\EN_1,\EN_2\in\linf$, $\gamma_1,\gamma_2\in (0,\infty)$
and
$\Sigma=\Sigma (\gamma_1,\gamma_2,\EN_1,\EN_2)$, the following two
statements are equivalent
\begin{enumerate}
\item[(a)] $\cagset\not=\emptyset$, and
\item[(b)] $\Sigma>0$.
\end{enumerate}
\end{proposition}
\begin{remark}
  The optimization problem above permits an interpretation in terms of
  the so-called \emph{optimal risk-sharing problem}. In case where the
  agents do not have access to a financial market, this problem
  has recently been addressed
  by many authors (see, e.g., \citet{BarElk04} for the exponential
  utility case, \citet{JouSchTou06} for monetary utility functionals
  and \citet{BarSca08} for concave preference functionals). When a
  financial market is present, the problem of
  optimal risk sharing when both agents have exponential utility has
  been studied in \citet{BarElk05}, where the authors focus on the form
  of the optimal structure.
\end{remark}
Before we proceed, we need to define several terms.
\begin{definition}\
\begin{enumerate}
\item The sum $\mathcal{E}=\mathcal{E}_{1}+%
  \mathcal{E}_{2}$ of the random endowments of the agents is called the
  \define{aggregate endowment}.
\item A pair $(B_1,B_2)$ in $(\linf)^2$ is called an
  \define{allocation}, while an allocation $(B_1,B_2)$ such that
  $B_1+B_2=\EN$ is called a \define{feasible allocation}; the set of
  all feasible allocations will be denoted by $F(\EN)$
\item For an allocation $(B_1,B_2)$, the sum
  $\unubone{B_1}+\unubtwo{B_2}$, denoted by $\sigma(B_1,B_2)$, is
  called the \define{score} of $(B_1,B_2)$. The difference
  $\sigma(B_1,B_2)-\sigma(\EN_1,\EN_2)$ is called the \define{excess
    score} (where, for simplicity, the parameters $\gamma_1$ and $\gamma_2$ are omitted from the notation).
\end{enumerate}
\end{definition}

By (\ref{equ:con-to-uncon}), the expression
$\nubtwo{B}-\nuwone{B}$ appearing in (\ref{equ:min-diff}) above
can be rewritten as
\begin{equation}
   \label{equ:rewrit}
   \begin{split}
     \nubtwo{B}-\nuwone{B} & =\nubtwo{B}+\nubone{-B}
= (\unubtwo{B+\EN_2}-\unubtwo{\EN_2})\\ & +
     (\unubone{-B+\EN_1}-\unubone{\EN_1}) = \sigma(\EN_1-B,
     \EN_2+B)- \sigma(\EN_1,\EN_2).
   \end{split}
\end{equation}
So that
\[ \Sigma(\gamma_1,\gamma_2,\EN_1,\EN_2)= \sup
\sets{\sigma(B_1,B_2)}{(B_1,B_2)\in F(\EN)}-\sigma(\EN_1,\EN_2).
\]
In words, $\Sigma$ is the maximized the excess score. If we
think of the aggregate endowment $\EN$ as the total wealth of our
two-agent economy, the solution of (\ref{equ:min-diff}) (if it
exists) will provide a redistribution of wealth so as to maximize
the (improvement in) the score.
Even though there is no direct economic reason why the sum of
individual indifference prices should be maximized, Proposition
\ref{pro:score-improve} - which is a mere restatement of the discussion
above - explains why the score is a useful concept.

\begin{proposition}
  \label{pro:score-improve}
  For each $B\in\linf$, the following two statements are equivalent:
  \begin{enumerate}
  \item $B\in\cagset$, and
  \item $\sigma(\EN_1-B,\EN_2+B)> \sigma(\EN_1,\EN_2)$.
  \end{enumerate}
\end{proposition}

The following proposition (compare to Theorem 2.3 in
\citet{BarElk05}) characterizes the score-optimal allocation.

\begin{proposition}
  For any $\EN_1,\EN_2\in\linf$ and $\gamma_1,\gamma_2>0$ there
  exists $B^*\in\linf$ such that
  \[ \sigma(\EN_1-B^*,\EN_2+B^*)\geq \sigma(B_1,B_2),\text{ for all }
  (B_1,B_2)\in F(\EN).\] Moreover, $B^*$ is unique up to
  replicability and
  \[ B^* \sim \frac{\gamma_1 \EN_1-\gamma_2 \EN_2}{\gamma_1+\gamma_2}.\]
\end{proposition}

\begin{proof}
  By (\ref{equ:rewrit}), it suffices to show that
  \begin{equation}
    \label{equ:rewritt}
    \begin{split}
      \unuwtwo{-B^*-\EN_2}+ \unuwone{B^*-\EN_1}\leq
      \unuwtwo{-B-\EN_2}+ \unuwone{B-\EN_1}
    \end{split}
\end{equation}
for all $B\in\linf$, which is, in turn, a consequence of Lemma
\ref{lem:B-1-2}. Indeed, it states that
\[
\unuwone{B-\EN_1} + \unuwtwo{-B-\EN_2} \geq \unuw{\tgamma}{ -
  \EN_1-\EN_2 },\]
with equality if and only if
\[ \tfrac{\gamma_1}{\tgamma} (B-\EN_1)\sim \tfrac{\gamma_2}{\tgamma}
(-B-\EN_2),\text{ i.e., } B \sim
\frac{\gamma_1\EN_1-\gamma_2\EN_2}{\gamma_1+\gamma_2 } .\]
\end{proof}

\begin{corollary}
\label{cor:agreement}
The following statements are equivalent:
\begin{enumerate}
\item $\cagset=\emptyset$,
\item $B^*= \frac{\gamma_1\EN_1-\gamma_2\EN_2}{\gamma_1+\gamma_2 }$ is
  replicable, and
\item $ \frac{\gamma_1}{\gamma_2} \EN_1\sim \EN_2$.
\end{enumerate}
\end{corollary}

\begin{remark}
We can relate the existence of mutually agreeable
non-replicable claims with the well-known notion of \emph{Pareto
optimality}. More precisely, an allocation
$(B_1,B_2)\in{F}(\EN)$ is called Pareto optimal if
$\nexists (C_1,C_2)\in F(\EN)$ s.t.
$B_i\preceq_{\gamma_i,\EN_i}C_i$ for $i=1,2$ and
$B_i\prec_{\gamma_i,\EN_i}C_i$ for at least one $i=1,2$. It
follows from Corollary \ref{cor:agreement}, that the
condition $\frac{\gamma_1}{\gamma_2} \EN_1\sim \EN_2$, implies
that the allocation $(\EN_1,\EN_2)$ is the unique (up to
replicability) Pareto optimal one. If
$\frac{\gamma_1}{\gamma_2} \EN_1\not\sim \EN_2$, a transaction
involving the optimal claim $B^*$ will lead to
 a Pareto optimal allocation.
\end{remark}

\bigskip

\section{Residual risk and an approximation of the indifference price}
\label{sec:approximation}
\subsection{Agreement and the residual risk process}
The notion of {\em residual risk} for the indifference
price valuation was defined in \citet{MusZar04a} for the setting of
our Example \ref{exa:first}., in \citet{StoZar06} for the
stochastic volatility model and in \citet{MusZar06a} for a
binomial-type model. This notion can be used to give a characterization of
contingent claims in $\agset$ (see Propositions
\ref{pro:agree-res-risk} and \ref{pro:agree-res-risk-qv} below).
Before we state this characterization, we give a short
introduction to residual risk in a static setting (see Appendix
\ref{sec:res-risk} for definitions and some properties in the
dynamic setting).
\subsubsection{Residual risk in a static setting}
Let $\gamma>0$ be a risk-aversion coefficient, and let $B\in\linf$
be a contingent claim. It can be shown that the optimization
problem with the value function
$\ufun{\gamma}{x}{B-\unuwg{B}}$, introduced in
(\ref{equ:value-function}) admits an essentially unique
maximizer $\bvt^{(B)}\in \bT$ (see Theorem \ref{thm:dual-rep} in Appendix A,
\citet{DelGraRheSamSchStr02}, page 105, Theorem 2.2 and \citet{KabStr02}, page 127, Theorem 2.1). The
corresponding wealth process
\[ X^{(B)}_t=\unuwg{B}+\int_0^t \bvt^{(B)}_u\, d\bS_u,\] can be
interpreted as the optimal risk-monitoring strategy for the writer
of the claim $B$, compensated by $\unuwg{B}$ at the initial time.
The hedging error\begin{equation}
   \label{equ:def-res-risk}
   \begin{split}
 \urrwg{B} = B-X^{(B)}_T
   \end{split}
\end{equation}
is called the \define{(writer's) residual risk}. $\urrwg{B}$
can be interpreted as the risk ``left in $B$'' after the optimal hedging
has been performed.
 Note that $\urrwg{B}=0$, a.s., for all replicable claims
 $B\in\linf$.
  In the conditional case, an analogous discussion
and the decomposition formula (see \eqref{equ:con-to-uncon})
\[ \nuwg{B}=\unuwg{B-\EN}-\unuwg{-\EN},\] allow us to define
the \define{conditional residual risk} $\rrwg{B}$ by
\[ \rrwg{B}=\urrwg{B-\EN}-\urrwg{-\EN},\]
and obtain the following decomposition
\begin{equation}\label{equ:dec-RR}
B= \nuwg{B}+\int_0^T  \bt^{(B|\EN)}_t\, d\bS_t+ \rrwg{B}
\end{equation}
where $\bt^{(B|\EN)}_t=\bt^{(B-\EN)}_t-\bt^{(-\EN)}_t$, $t\in
[0,T]$. The process $\prf{\bt^{(B|\EN)}_t}$, as well as the
decomposition (\ref{equ:dec-RR}), could have been derived
equivalently using the optimization problems used to define the
conditional indifference prices. All of the above concepts have
natural analogues when seen from the buyer's side. Namely, we
define the (buyer's) residual risk by
$R^{(b)}(B;\gamma)=\urrwg{-B}$ and by
$R^{(b)}(B;\gamma|\EN)=R^{(b)}(B+\EN;\gamma)-R^{(b)}(\EN;\gamma)$
in the conditional case.
\begin{remark}
  Using decomposition (\ref{equ:dec-RR}) and the Proposition
  \ref{pro:risk-measures}, we observe
  that \[\nubg{R^{(w)}(B;\gamma|\EN)}=\nubg{B}-\nuwg{B}.\] Hence, the
  agreement condition (\ref{agreementcondition}) can also be written
  as \[B\in\agset\Leftrightarrow\nubg{R^{(w)}(B;\gamma|\EN)}\geq 0.\]
  In particular, $\nubg{R^{(w)}(B;\gamma|\EN)}< 0$ for all
  $\EN\in\mathcal{R}^{\infty}$, and
  $B\in\linf\setminus\mathcal{R}^{\infty}$.  We also note that
  $\nuwg{\rrwg{B}}=0$, for all  $\EN,B\in\linf$.
\end{remark}
The following proposition gives
a characterization of mutually
agreeable contingent claims in terms of their residual risk.

\begin{proposition}
\label{pro:agree-res-risk} For
$B\in\linf$, the following statements are equivalent:
\begin{enumerate}
\item $B\in\agset$,
\item the inequality
\begin{equation}
\label{equ:dif-res}
\EE^{\QQ}\left[ \rrw{\gamma_1}{B}{\EN_1} \right]+\EE^{\QQ}\left[ \rrb{\gamma_2}{B}{\EN_2} \right]
\geq 0
\end{equation}
holds for some $\QQ\in \mathcal{M}_{e,f}$.
\item the inequality
\begin{equation}
\label{equ:dif-res-two}
  \EE^{\QQ}\left[ \urrw{\gamma_1}{B-\EN_1}-\urrw{\gamma_1}{-\EN_1}
  \right]
+
 \EE^{\QQ}\left[ \urrw{\gamma_2}{-B-\EN_2}-\urrw{\gamma_2}{-\EN_2}
  \right]\geq 0
\end{equation}
holds for some $\QQ\in \mathcal{M}_{e,f}$.
\item the inequality \eqref{equ:dif-res} holds for all
  $\QQ\in\MM_{e,f}$.
\item the inequality \eqref{equ:dif-res-two} holds for all
  $\QQ\in\MM_{e,f}$.
\end{enumerate}
\end{proposition}
\begin{proof}
It suffices to make the following two observations
\begin{itemize}
\item[(a)]  thanks to the definition \eqref{equ:def-res-risk} of
residual risk, the differences $ \nuwone{B} -
\rrw{\gamma_1}{B}{\EN_1}$ and
$\nubtwo{B}-\rrb{\gamma_2}{B}{\EN_2}$ are both of the form
$\int_0^T \bt_t\, d\bS_t$ with $\bt\in\bT$, and \item[(b)] the
following equality holds
\begin{equation}
    \nonumber
   \begin{split}
  \rrw{\gamma_1}{B}{\EN_1}-\rrb{\gamma_2}{B}{\EN_2}&=
\urrw{\gamma_1}{B-\EN_1}-\urrw{\gamma_1}{-\EN_1} \\ &+
\urrw{\gamma_2}{-B-\EN_2}-\urrw{\gamma_2}{-\EN_2}.
   \end{split}
\end{equation}
\end{itemize}
\ \\[-6ex]
\end{proof}
\begin{remark}
  It should be pointed out that it is enough to check the above
  inequalities just for some probability measure in
  $\mathcal{M}_{e,f}$. Also, it follows from the definition of the
  residual risk  that the inequality \eqref{equ:def-res-risk} implies
  that the transaction of claim $B$ at any price $p$ will decrease the
  sum of expected residual risks. If in addition $p\in
  (\nuwone{B},\nubtwo{B})$, each agent's expected residual risk will
  be decreased.
\end{remark}

Under the additional mild assumption of left-continuity for the
filtration $\mathbb{F}$, we can replace the criterion given in
Proposition \ref{pro:agree-res-risk} by the following one (see
Appendix \ref{sec:res-risk} for the additional notation), which sometimes is easier to check.
\begin{proposition}
\label{pro:agree-res-risk-qv}
Suppose that $\mathbb{F}$ is continuous. For $B\in\linf$, the
following two statements are equivalent:
\begin{enumerate}
\item $B\in\agset$, and
\item
\begin{multline*}
\gamma_{1} \EE^{\QZ}\Big[
\Bab{\urrw{\gamma_1}{B-\EN_1}}_T-
\Bab{\urrw{\gamma_1}{-\EN_1}}_T
\Big]+\\
+
\gamma_{2} \EE^{\QZ}\Big[
\Bab{\urrw{\gamma_2}{-B+\EN_2}}_T-
\Bab{\urrw{\gamma_2}{\EN_2}}_T
\Big] \geq 0
\end{multline*}
\end{enumerate}
\end{proposition}
\begin{proof}
  The equivalence  follows from Proposition \ref{pro:agree-res-risk} and
  part (2) of Theorem \ref{thm:man-sch}, which effectively state that
  $\bab{\urrwg{B}}_t=\bab{\ullwg{B}}_t$ for all $t\in [0,T]$, so that
  $\turrwg{B}-\frac{\gamma}{2}\bab{\turrwg{B}}_T$ is a
  $\QZ$-martingale, for any $\gamma>0$ and $B\in\linf$.
\end{proof}

\begin{example}
\label{exa:first}
This example is set in an incomplete financial market similar to the
one considered in \citet{MusZar04a} (see, also, \citet{Hen02}).
The market consists of one risky asset
$S=\left( S_{t}\right) _{t\in [0,T]}$ with dynamics
\begin{eqnarray*}
dS_{t} &=&S_t\big(
\mu (t)\,dt+\sigma (t)\,dW_{t}^{(1)}\big)
\end{eqnarray*}%
and an additional (non-traded) factor
$Y=\left( Y_{t}\right) _{t\in [0,T]}$ which evolves is a unique strong
solution of
\begin{eqnarray*}
dY_{t} &=&b(Y_{t},t)\,dt+a(Y_{t},t)\left( \rho dW_{t}^{(1)}+\rho
^{\prime }dW_{t}^{(2)}\right),
\end{eqnarray*}%
where $W^{(1)}=(W_{t}^{(1)})_{t\in \lbrack 0,T]}$ and
$W^{(2)}=(W_{t}^{(2)})_{t\in \lbrack 0,T]}$ are two standard
independent Brownian Motions defined on a probability space $(\Omega
,\mathbb{F},(\mathcal{F}_{t}),\PP)$. The constant $\rho\label{corr.}
\in (-1,1)$ is the correlation coefficient and $\rho ^{\prime
}\label{rho.prime}:=\sqrt{1-\rho ^{2}}$. We assume that the
deterministic functions $\mu $, $\sigma :[0,T]\rightarrow \R $\ are
uniformly bounded ($\sigma >0$).

By Theorem 3 in \citet{MusZar04a} (in \citet{MusZar04a} $\mu $ and $\sigma $ are
constants, but their  arguments carry over to our setting), we have that
\begin{equation}\label{ex1.formula}
v^{(w)}\left( B;\gamma \right) =\frac{1}{\gamma \left( 1-\rho ^{2}\right) }%
\ln \left\{ \EE^{\QZ}\left( e^{\gamma \left( 1-\rho ^{2}\right)
B}\right) \right\},
\end{equation}
for any payoff
$B\in\linf$, such that $B=g(Y_T)$ for some bounded Borel  function
$g:\R\to\R$,
where the Radon-Nikodym derivative of $\QZ$ is given by
\begin{equation*}
  \frac{d\QZ}{d\PP}=\exp\left(\int_{0}^{T }\frac{1}{2}\lambda ^{2}(t)dt+\int_{0}^{T }\lambda
    (t)dW_{t}^{(1)}\right),
\end{equation*}
and $\lambda (t):=\frac{\mu (t)}{\sigma (t)}$ is the
\textit{Sharpe ratio} of $S$.
\end{example}

Let us further suppose that $\EN_1=g_1(Y_T)$ and
$\EN_2=g_2(Y_T)$ for some Borel bounded functions $g_1$ and $g_2$.
Proposition \ref{pro:agr-prices} and representation
(\ref{ex1.formula}) imply that  that $B=g(Y_T)\in\agset$ if and only if
\begin{equation*}
\left( \frac{\EE^{\QZ}\left( e^{\gamma _{1}\tilde{B}}e^{-\gamma _{1}\tilde{\EN}%
_{1}}\right) }{\EE^{\QZ}\left( e^{-\gamma _{1}\tilde{\EN}_{1}}\right) }\right) ^{%
\frac{\gamma _{2}}{\gamma _{1}}}\leq \frac{\EE^{\QZ}\left( e^{-\gamma _{2}\tilde{%
\EN}_{2}}\right) }{\EE^{\QZ}\left( e^{-\gamma _{2}\tilde{\EN}_{2}}e^{-\gamma _{2}%
\tilde{B}}\right) }
\end{equation*}
where $\tilde{B}=(1-\rho^2)B$ and $\tilde{\EN_i}=(1-\rho^2)\EN_i$,
$i=1,2$.

As we have seen, $\unuwone{B}> \unubtwo{B}$, $\forall B\nsim 0$.  It is
easy to verify that $\nuwone{B}\leq\unuwone{B}$ if and only if
$\cov^{\QZ}(\EN_1,B)\geq 0$, where $\cov^{\QZ}(.,.)$ is the covariance
under the measure $\QZ$. This means that the presence of random
endowment which is positively correlated with the claim payoff,
reduces the writer's indifference price. Similarly,
$\nubtwo{B}\geq\unubtwo{B}$ if and only if $\cov^{\QZ}(\EN_2,B)\leq
0$. Therefore, we infer that a necessary condition for a claim $B$ to be
mutually agreeable is that $\cov^{\QZ}(\EN_1,B)> 0$ or
$\cov^{\QZ}(\EN_2,B)< 0$.

\subsection{An asymptotic approximation of indifference prices}
When the size of the claim whose price is negotiated is small compared
to the sizes of agent's contingent claims (and this is typically the
case in practice), one can use a Taylor-type
expansion of the indifference price around $0$, and obtain
more precise quantitative answers to the agreement question. More
precisely, we assume that claim under consideration has the form
$\ba\cdot \bB$ for a some vector $\bB=(B_1,\dots, B_n)$ in
$(\linf)^n$, where $\ba\in\R^n$ should be interpreted as a small
parameter.
We start with a single agent's point of view and present an
approximation result for the indifference price $\nuwg{\ba\cdot\bB}$
of $\ba\cdot\bB$, where $\EN\in\linf$ and $\gamma>0$. To alleviate the
notation, we shorten $\nuwg{\ba\cdot\bB}$ to $w(\ba)$, for $\ba\in\R^n$.

A straightforward extension of Theorem 5.1 on page 590 in
\citet{IlhJonSir05}, where we use the fact that the conditional
indifference prices are just the conditional ones under a changed
measure, yields the following result:
\begin{proposition}
\label{pro:first-der}
  The function $w$ is continuously differentiable on $\R^n$ with
\[ \nabla w(\ba) = \EE^{\hq{\gamma \ba\cdot\bB-\gamma\EN}} [\bB]=
\big( \EE^{\hq{\gamma \ba\cdot \bB-\gamma\EN}}[B_1],\dots,
\EE^{\hq{\gamma \ba\cdot \bB-\gamma\EN}}[B_n] \big), \
\ba\in\R^n.\]
\end{proposition}
The concept of minimized variance, defined below, is important for
the study of second derivatives of the function $w$.
\begin{definition} Let $\QQ\in\MM_e$ be an arbitrary martingale
  measure.
\label{def:controlled-variance}
\begin{enumerate}
\item For $B\in\linf$ we define the \define{projected variance}
$\Delta^{\QQ}(B)$ of $B$ under $\QQ$ as :
\begin{equation}
\Delta^{\QQ}(B)=\inf_{\bt\in\bT^2_{\QQ}} \EE^{\QQ}\left[ (
  B-\EE^{\QQ}[B]- (\bt\cdot \bS)_T)^2\right],
\end{equation}
where, $\bT^2_{\QQ}=\{\bt\in L(\bS):(\bt\cdot\bS)\text{ is square
  integrable martingale under }\QQ\}$, so that
$\underset{\QQ\in\mathcal{M}_{e,f}}{\bigcap}\bT^2_{\QQ}\subset\bT$.
\item For $B_1,B_2\in\linf$  we define the \define{projected
covariance} $\Delta^{\QQ}(B_1,B_2)$ of
  $B_1$ and $B_2$ by polarization:
\[ \Delta^{\QQ}(B_1,B_2)= \tot (
\Delta^{\QQ}(B_1+B_2)-\Delta^{\QQ}(B_1)-\Delta^{\QQ}(B_2)).\]
\item For a vector $\bB=(B_1,\dots, B_n)\in (\linf)^n$ and a probability
measure $\QQ\in\MM_e$,  we define the
\define{$\QQ$-projected variance-covariance} matrix $\bD^{\QQ}(\bB)$ by
\[ \bD^{\QQ}_{ij}(\bB)= \Delta^{\QQ} (B_i, B_j),\ i,j=1,\dots, n.\]
\end{enumerate}
\end{definition}

\begin{remark}\
\begin{enumerate}
\item The projected variance $\Delta^{\QQ}(B)$ is the square of the
  $\ltwo(\QQ)$-norm of the projection $P^{\QQ}(B)$ of the random
  variable $B\in\linf\subseteq \ltwo(\QQ)$ onto the closed subspace
  $\R\oplus \sets{ (\bt\cdot\bS)_T}{\bt\in \bT^2_{\QQ}}$ of
  $\ltwo(\QQ)$ (the closeness of $\sets{ (\bt\cdot\bS)_T}{\bt\in
    \bT^2_{\QQ}}$ in $\ltwo(\QQ)$ is an immediate consequence of the
  $\ltwo(d[\bS])$-$\ltwo(\QQ)$ isometry of stochastic
  integration). It follows that the projected covariance
  $\Delta^{\QQ}(B_1,B_2)$ can be represented as
\[ \Delta^{\QQ}(B_1,B_2)= \EE^{\QQ}[ P^{\QQ}(B_1) P^{\QQ}(B_2)].\]
In particular, $\Delta^{\QQ}(\cdot,\cdot)$ is a bilinear functional on
$\linf\times\linf$ and the following equality holds
\begin{equation}
   \label{equ:bilinear}
   \begin{split}
 \Delta^{\QQ} (\ba\cdot \bB)=\ba\cdot \bD^{\QQ}(\bB) \ba=\sum_{i,j=1}^n \alpha_i
\bD^{\QQ}_{ij}(\bB) \alpha_j,
   \end{split}
\end{equation}
for all $\QQ\in\MM_e$, $\bB\in(\linf)^n$ and $\ba=(\alpha_1,
\dots, \alpha_n)\in\R^n$.
\item Details on the notion of the
projected variance, which is closely related to  mean-variance
hedging, can be found in \citet{FolSon86} or \citet{Sch01a}. Note
that the existence of the minimizer in the Definition
\ref{def:controlled-variance} for bounded claims can be established using
Kunita-Watanabe decomposition of the uniformly integrable
$\QQ$-martingale $\left(B_{t}\right)_{t\in \left[
    0,T\right]}$ defined as $B_{t}=\EE^{\QQ}[
  B|\mathcal{F}_{t}]$. For details on the Kunita-Watanabe
decomposition we refer the reader to \citet{AnsStr93b}.
\end{enumerate}
\end{remark}
We remind the reader that a filtration $\mathbb{F}=\prf{\FF_t}$ is said to be
\define{left-continuous} if $\FF_{t}=\sigma(\cup_{s<t} \FF_s)$, for all $t\in
(0,T]$.
\begin{lemma}
\label{lem:sec-der}Suppose
that $n=1$ and that the filtration $\mathbb{F}$ is
left-continuous. Then for $B,\EN \in\linf$, the function
$w:\R\to\R$  is
twice differentiable at any $\alpha\in\R$ and its second
derivative is given by
\begin{equation}
w'' ( \alpha ) =\gamma \Delta^{\hq{\gamma\alpha B-\gamma\EN}}( B).
\end{equation}
\end{lemma}

\begin{proof} Without loss of generality we may suppose that $\EN=0$
- otherwise, we just work under the measure
  $\PP_{-\gamma\EN}$. Let us focus on the case $\alpha=0$, where first derivative
  $w'(0)$ is equal to $\EE^{\QZ}[B]$ by Proposition
  \ref{pro:first-der}, so a simple translation by a constant (under
  which $\Delta$ is clearly invariant) allows us
  to assume, in addition, that $\EE^{\QZ}[B]=0$. It will be enough,
  therefore, to show that
\[ \lim_{\alpha\to 0} \abs{ \frac{\unuwg{\alpha B}}{\alpha^2} - \tfrac{\gamma}{2}
  \Delta^{\QZ}(B)}=0.\]
By the sign invariance of $\Delta(\cdot)$ and the scaling property
(\ref{equ:scaling-property}) of the indifference prices, it suffices
to consider only $\alpha>0$, i.e., it is enough to establish that
\begin{equation}
   \label{equ:sec-enough}
   \begin{split}
     \lim_{\alpha\searrow 0} \abs{ \frac{\unuw{\alpha \gamma} {
           B}}{\alpha} - \tfrac{\gamma}{2} \Delta^{\QZ}(B)}=0
   \end{split}
\end{equation}
Theorem \ref{thm:man-sch} and the definition of the residual risk in
Appendix \ref{sec:res-risk} state that
\[ \frac{\unuw{\alpha \gamma}{B}}{\alpha} = \frac{\gamma}{2}
\EE^{\QZ}[ \ab{\ullw{\alpha\gamma}{B}}_T]=\frac{\gamma}{2} \EE^{\QZ}[
  \Tullw{\alpha\gamma}{B} ^2],\]
where $\prf{\tullw{\alpha\gamma}{B}}$ is as in Theorem
\ref{thm:man-sch}. The BMO-convergence of the processes
$\prf{\tullw{\alpha\gamma}{B}}$ from the same theorem, implies, in
particular, the $\ltwo(\QZ)$-convergence of their terminal values,
i.e.,
\[ \Tullw{\alpha\gamma}{B}\to \Tullw{0}{B}\text{ in } \ltwo(\QZ).\]
Therefore, it remains to prove that
\[ \EE^{\QZ}\left[\Tullw{0}{B}^2\right]= \Delta^{\QZ}(B),\ \forall\, B\in\linf.\]
Thanks to the final part of Theorem \ref{thm:man-sch},
$\ullw{0}{B}$ is strongly orthogonal to any process of the form
$(\bt\cdot\bS)$, for $\bt\in \bT^{2}_{\QZ}$. In particular, for
$\bt\in\bT^{2}_{\QZ}$ and $\hat{\bt}^{(B)}$ as in Theorem
\ref{thm:man-sch},  we have
\[ B-(\bt\cdot\bS)_T= \big((\hat{\bt}^{(B)}-\bt)\cdot
\bS\big)_T+\Tullw{0}{B},\]
so that
\[ \EE^{\QZ}\left[ (B-(\bt\cdot\bS)_T)^2\right]= \EE\left[\Big(
  \big((\hat{\bt}^{(B)}-\bt)\cdot \bS\big)_T\Big)^2\right]+\EE\left[
  \Tullw{0}{B}^2\right].\]
Consequently, the minimum in the definition of $\Delta^{\QZ}[B]$
is attained at $\bt=\hat{\bt}^{(B)}$ so that
$\Delta^{\QQ}(B)=\EE^{\QZ}[ \Tullw{0}{B}^2 ]$, which is the
equality we set out to prove.

For $\alpha\neq 0$, we may again suppose that
$w'(\alpha)=\EE^{\QQ^{(\gamma\alpha B)}}[B]=0$. Hence, it is
enough to show that $\lim_{\eps\to
0}\frac{w(\alpha+\eps)-w(\alpha)}{\eps^2}=\frac{\gamma}{2}\Delta^{\QQ^{(\alpha\gamma
B)}}(B)$. For this, we note that
$w(\alpha+\eps)-w(\alpha)=v^{(w)}(\eps B|-\alpha B;\gamma)$ i.e.,
we can rewrite the second derivative at some $\alpha\neq 0$ as the
second derivative at $\alpha =0$ with random endowment. This observation
finishes the proof.
\end{proof}

The case $n>1$ is covered by the following lemma.
\begin{lemma}
\label{lem:taylor}
For $\ba,\bdelta\in\R^n$, $\bB=(B_1,\dots, B_n)\in(\linf)^n$ and
$\EN\in\linf$, we have
\begin{equation}
   \label{equ:limit-delta}
   \begin{split}
 \lim_{\eps\to 0} \frac{w(\ba+\eps \bdelta)-w(\ba)-\eps \nabla w(\ba)\cdot
  \bdelta}{\eps^2}= \tot \sum_{i,j=1}^n \bdelta_i
\bD^{\QQ}_{ij}(\bB) \bdelta_j.
   \end{split}
\end{equation}

\end{lemma}
\begin{proof}
  We can assume that $\ba=(0,\dots, 0)$ by absorbing the term
  $-\ba\cdot \bB$ into the random endowment $\EN$. The left-hand side
  of (\ref{equ:limit-delta}) can now be understood as the second
  derivative at $0$ of the function $\tilde{w}:\R\to \R$ given by
\[ \tilde{w}(\eps)=\nuwg{\eps \bdelta\cdot \bB}.\]
We can finish the proof by employing Lemma \ref{lem:sec-der} and using
the equality (\ref{equ:bilinear}) with $\QQ=\hq{-\gamma\EN}$ and
$\bdelta$ substituted for $\ba$.
\end{proof}

With the above results in  our toolbox,
we can give a second order directional Taylor-type
approximation of the indifference price.
\begin{proposition}
\label{pro:approx}
Choose
 $\bB\in (\linf)^n$, $\ba\in \R^n$, $\gamma>0$ and $\EN\in\linf$, and assume
 that the filtration $\mathbb{F}$ is left-continuous. With the notions
 of projected variance and covariance as in Definition
\ref{def:controlled-variance}, we have the following equality
\begin{equation}
   \label{equ:approx}
   \begin{split}
     \nuwg{\eps\ba\cdot \bB}= \eps\ba \cdot \EE^{\hq{-\gamma\EN}}[ \bB ]+
     \frac{\eps^{2}\gamma}{2}\ba \cdot
     \bD^{\hq{-\gamma\EN}}(\bB)\ba+o(\eps^2),\text{ as }\eps \to 0.
   \end{split}
\end{equation}
\end{proposition}
Although the asymptotic expansion (\ref{equ:approx}) in Proposition
\ref{pro:approx} above is important in its own right, its main
application is the following criterion for mutual agreement for a
small quantity of a given contingent claim.

\begin{proposition}\label{pro:approx-agree}
Suppose that $\mathbb{F}$ is left-continuous and that the random
endowments $\EN_1,\EN_2\in \linf$ and risk-aversion coefficients
$\gamma_1,\gamma_2>0$ are chosen. Let $B\in\linf$ be a given
contingent claim. The set $\cagset$ contains a segment of the form
\[ \sets{ \alpha B}{\alpha \in (0,\alpha_0)}\text{ for some }
\alpha_0>0 \text{ if and only if }\EE^{\hq{-\gamma_1 \EN_1}}[B]< \EE^{\hq{-\gamma_2 \EN_2}}[B],\]
and similarly, the set $\cagset$ contains a segment of the form
\[ \sets{ \alpha B}{\alpha \in (-\alpha_0,0)}\text{ for some }
\alpha_0>0 \text{ if and only if },\EE^{\hq{-\gamma_1 \EN_1}}[B]> \EE^{\hq{-\gamma_2 \EN_2}}[B].\]
\end{proposition}

\begin{proof}
We first note that the convexity of $\cagset$ implies that if $\exists$ $\alpha_0>0$ such that $\alpha_0B\in\cagset$, then
$\alpha B\in\agset$, $\forall \alpha\in[0,\alpha_0]$.
By equation (\ref{pro:approx}), $\nuwone{\alpha B}= \alpha\EE^{\hq{-\gamma_1\EN_1}}[ B ]+o(\alpha)$ and
$\nubtwo{\alpha B}= \alpha\EE^{\hq{-\gamma_2\EN_2}}[ B ]+o(\alpha)$. Hence, the inequality $\EE^{\hq{-\gamma_1 \EN_1}}[B]< \EE^{\hq{-\gamma_2 \EN_2}}[B]$
yields that there exists $\alpha_0 >0$ small enough such that $\nuwone{\alpha_0 B}<\nubtwo{\alpha_0 B}$, i.e., $\alpha_0 B\in\cagset$.

On the other hand, suppose that there exists $\alpha_0>0$ such that $\alpha_0 B\in\cagset$ and assume that
$\EE^{\hq{-\gamma_1 \EN_1}}[B]\geq \EE^{\hq{-\gamma_2 \EN_2}}[B]$. It is easy to check that $\Delta^{\hq{-\gamma_i\EN_i}}(B)>0$ for $i=1,2$
(since $B\notin\mathcal{R}^{\infty}$) and that by (\ref{pro:approx}) and its buyer's version, we get
\[ \alpha\EE^{\hq{-\gamma_1 \EN_1}}[B]+\frac{\alpha^{2}\gamma_1}{2}\Delta^{\hq{-\gamma_1\EN_1}}(B)<
\alpha\EE^{\hq{-\gamma_2 \EN_2}}[B]-\frac{\alpha^{2}\gamma_2}{2}\Delta^{\hq{-\gamma_2\EN_2}}(B)+o(\alpha^2)\]
for every $\alpha$ close to zero such that $0<\alpha\leq\alpha_0$ (note that thanks to the linearity of $\bT^{2}_{\QQ}$, we have $\Delta^{\QQ}(B)=\Delta^{\QQ}(-B)$, $\forall B\in\linf$).

This gives that for any such $\alpha$
\[\frac{\alpha^2}{2}(\gamma_1\Delta^{\hq{-\gamma_1\EN_1}}(B)
+\gamma_2\Delta^{\hq{-\gamma_2\EN_2}}(B))+o(\alpha^2)<0, \] Dividing
through by $\alpha^2$ and letting $\alpha\to 0$, we get that
$\gamma_1\Delta^{\hq{-\gamma_1\EN_1}}(B)+\gamma_2\Delta^{\hq{-\gamma_2\EN_2}}(B)\leq
0$, which is a contradiction.  The proof of the second argument is
similar and hence omitted.
\end{proof}

\begin{remark}\label{rem:approx-interval}
Proposition \ref{pro:approx-agree} can be used to provide an approximation of the size of the set of
the agreement prices for small number of units. More precisely, if $\EE^{\hq{-\gamma_1 \EN_1}}[B]\neq \EE^{\hq{-\gamma_2 \EN_2}}[B]$, it holds that
\begin{multline}
\label{equ:approx-interval}
\nubtwo{\alpha B}-\nuwone{\alpha B}=\\ =\alpha(\EE^{\hq{-\gamma_2
    \EN_2}}[B]-\EE^{\hq{-\gamma_1 \EN_1}}[B])-
\frac{\alpha^{2}}{2}(\gamma_1\Delta^{\hq{-\gamma_1\EN_1}}(B)+\gamma_2\Delta^{\hq{-\gamma_2\EN_2}}(B))+o(\alpha^2)
\end{multline}
for every $\alpha\in\R$ close to zero such that $\alpha B\in\cagset$. Another application of Proposition \ref{pro:approx-agree}
is given at the first part of Remark \ref{rem:pepq}.
\end{remark}

A second order approximation of the optimal strategy has been
recently given in \citet{KraSir06} for general utilities defined on
the positive real line. For the Example \ref{exa:first}, the
corresponding approximation result is easily obtained for the
unconditional case by the formula (\ref{ex1.formula}) (see, e.g.,
\citet{Hen02}). After changing the measure $\PP$ to
$\PP_{-\gamma\EN}$, it is straightforward to show that
\begin{equation*}
\nuwg{\alpha
B}=\alpha\EE^{\QQ^{(-\gamma\EN)}}[B]+\frac{\alpha^2}{2}\gamma
(1-\rho^2) \Var^{\QQ^{(-\gamma\EN)}}(B)+o(\alpha^2)\text{, }\forall
B\in\linf,
\end{equation*}
where $\var^{\QQ}(B)$ denotes the variance of random variable $B$
under the probability measure $\QQ$.

\begin{example}
\label{exa:second}
In the cases where there is no closed-form expression for the indifference price, the approximation
(\ref{equ:approx}) is rather useful. One of these cases is the
stochastic volatility model studied in \citet{SirZar05} and \citet{IlhJonSir04} (see also \citet{Hed07}). Proposition
\ref{pro:approx-agree} is then a generalization of the approximation results of
\citet{IlhJonSir04}, subsection 4.3.
\end{example}

\bigskip

\section{Partial equilibrium prices}\label{sec:pepq}
This section deals with the existence and the uniqueness of a
partial equilibrium price of a contingent claim in the simplified
two-agent economy. The discussion of
  mutual agreeability
in  previous sections assumed that
 the number of units is fixed and the claim is indivisible.
If, however, the negotiation between agents involves the quantity
traded as well as the price, and if this quantity is not constrained
by quantization, a great deal more can be said about the outcome of
the negotiation. The main advantage is that the methodology of
equilibrium theory  can be applied and a unique price-quantity pair singled
out on the basis of the fundamental economic principle of market
clearing.

\subsection{The partial equilibrium}
A vector $\bB=(B_1,\dots, B_n)\in (\linf)^n$ of contingent claims is
chosen and kept constant throughout this section.  The extended set
$\R\cup\set{\pm\infty}$ of real numbers is denoted by $\eR$.

It will
be notationally convenient to define the ``restrictions''
 $\bU_i:\eR^n\times \R^n \to \R$,
$i\in\set{1,2}$ of the value functions $\uone{\cdot}$ and
$\utwo{\cdot}$ (see (\ref{equ:value-function})) by
\begin{equation}
   \label{equ:bug}
   \begin{split}
\bU_i(\ba;\bpr)=
\begin{cases}
 \ui{\ba\cdot(\bB-\bpr)}, & \ba\in \R^n, \\
\limsup_{\ba'\to\ba,\, \ba'\in\R^n} \bU_i(\ba';\bpr), &
\ba\in \eR^n\setminus \R^n,
\end{cases}
   \end{split}
\qquad i\in\set{1,2},\  \bpr\in\R^n.
\end{equation}
i.e., $\bU_i(\cdot;\bpr)$ is the extension of the continuous function
$\bU_i(\cdot;\bpr)\Big|_{\R^n}$ to $\eR^n$ by upper semi-continuity
and gives the indirect utility of agent $i$ when she holds $\ba$ units
of $\bB$, purchased at price $\bpr$.
\begin{definition}\label{def:demand}
The \define{demand correspondence}
$Z_i: \R^n \to 2^{\eR^n}$, for the agent $i\in\set{1,2}$, is defined by
\begin{equation}
   \label{equ:demand-fct}
   \begin{split}
Z_i(\bpr) =\argmax\sets{ \bU_i(\ba,\bpr) }{ \ba\in \eR^n},\ \bpr\in\R^n.
   \end{split}
\end{equation}
\end{definition}
Intuitively, the elements of $Z_i(\bpr)$ give the numbers of units of
$\bB$ that agent $i$ is willing to purchase at price $\bpr$ (the
numbers of units that maximize her indirect utility).  Using the above
notation and definition, we are ready to introduce the central concept
of the section:
\begin{definition}\label{def:pepq}
A pair $(\bpr,\ba)\in\R^n\times\R^n$ is called a \define{partial-equilibrium
  price-quantity (PEPQ)} if
\begin{equation}
   \label{equ:pepq-condition}
   \begin{split}
\ba\in Z_1(\bpr)\text{ and } -\ba\in Z_2(\bpr).
   \end{split}
\end{equation}
A vector $\bpr\in\R^n$ for which  there exists $\ba\in\R^n$ such that
$(\bpr,\ba)$ is a PEPQ is called a \define{partial-equilibrium price (PEP)}.
\end{definition}
In other words, the PEP is the price-vector of the contingent claim
$\bB$ at which the quantity that one agent is willing to sell is equal
to the quantity which the other agent is wants to buy.

When there exists $\ba\in\R^m$ such that the
 contingent claim $\ba\cdot \bB$ is replicable,
any PEP $\bpr$ must have the property that $\ba\cdot \bpr= p_{NA}$,
where $p_{NA}$ is the replication price of $\ba\cdot\bB$; this will
hold no matter what the characteristics of the agents are.
 It is, therefore,
only reasonable to assume that such claims do not enter the
negotiation, i.e., we enforce
the following assumption for the remainder of the section:
\begin{assumption}
\label{ass:no-repl}
  There exists no $\ba\in\R^n\setminus\set{\bz}$
such that $\ba\cdot
  \bB\sim 0$.
\end{assumption}

\subsection{Properties of the demand functions}
Let $\Pna\subseteq \R^n$ be the set of all arbitrage-free price-vectors
of the contingent claims $\bB$, i.e.,
\[ \Pna=\sets{ \EE^{\QQ}[\bB] }{ \QQ\in\MM_e},\]
where, as usual, $\EE^{\QQ}[\bB]=(\EE^{\QQ}[B_1], \dots ,
\EE^{\QQ}[B_n])\in\R^n$.
To simplify the notation in the sequel, we introduce two
$n$-dimensional families of measures in $\MM_e$, parametrized by
$\ba\in\R^n$:
\[ \hqi{\ba}= \hq{\gamma_i\ba\cdot\bB-\gamma_i\EN_i},\ \ba\in\R^n,\ i=1,2.\]
As we will see below (see Proposition \ref{pro:equ-char}), when we
are looking for partial equilibrium prices, we can restrict
ourselves to the sets
\[ \Pui=\sets{ \EE^{\hq{\ba}_i}[\bB]}{\ba\in\R ^n},\text{ for } i=1,2.\]

In general, $\Pui\subseteq \Pna$, for $i=1,2$.  The equality holds
when $\EN_i\sim 0$ (see \citet{IlhJonSir05}, Lemma 7.1). For
future use, we define the function $\bu_i
:\R^n\to\overline{\R}$ by  $\bu_i(\bpr)=\underset{\ba\in\bar{\R}^{n}}{\sup
}\left\{ \bU_{i}(\ba ;\bpr)\right\}$, for $i\in\set{1,2}$.
Building on the notation of Section \ref{sec:approximation}, we
also introduce the following two shorthands:
\begin{equation}
   \label{equ:w-and-b}
   \left.\begin{split}
w_i(\ba)&=\nuwi{\ba\cdot \bB}\\
b_i(\ba)&=\nubi{\ba\cdot \bB}
   \end{split}\right\}
\quad \ba\in\R^n,\, i\in\set{1,2}.
\end{equation}
\begin{lemma}
  \label{lem:strict-conv} For $i=1,2$,
  $w_i$ is strictly convex and $b_i$ is strictly concave.
\end{lemma}
\begin{proof}
A change of measure argument (where we replace $\PP$ by
$\PP_{-\gamma_i\EN_i}$) can be employed to justify no loss of
generality if we assume that $\EN_i=0$ in this proof.
The fact that $w_i(\cdot)$
is convex follows from the convexity of the indifference price. In
order to establish that the convexity is, in fact, strict, we
assume, to the contrary, that there exist $\ba_1,\ba_2\in \R^n$
with $\ba_1\neq\ba_2$ and $\ld\in (0,1)$ such that
\[ w_i(\ld \ba_1+(1-\ld)\ba_2)=\ld w_i(\ba_1)+(1-\ld) w_i(\ba_2).\]
Equivalently, we have
\begin{equation}
    \nonumber
   \begin{split}
\unuw{\gamma_i}{(\ld\ba_1+(1-\ld) \ba_2) \cdot \bB}&=
\unuw{\tfrac{\gamma_1}{\ld}}{\ld \ba_1\cdot \bB}+
\unuw{\tfrac{\gamma_1}{1-\ld}}{(1-\ld) \ba_2\cdot \bB}
   \end{split}
\end{equation}
Since
$
(\tfrac{\gamma_i}{\ld})^{-1}+
(\tfrac{\gamma_i}{1-\ld})^{-1}=(\gamma_i)^{-1}$,
we can use Lemma \ref{lem:B-1-2} to conclude that
\[  \ba_1\cdot \bB\sim \ba_2\cdot \bB, \text{ i.e. } \ba\cdot \bB\sim
0,\text{ where $\ba=\ba_1-\ba_2\not= 0\in\R^n$,}\] a contradiction
with Assumption \ref{ass:no-repl}. A similar argument can be
employed to prove strict concavity of $b_i$, $i\in\set{1,2}$.
\end{proof}

\begin{proposition}\label{pro:demand}
For $i\in\set{1,2}$, the  functions
  $\bu_i(\cdot)$ and $Z_i(\cdot)$ have
  the following properties
  \begin{enumerate}
  \item The maximum in \eqref{equ:demand-fct} is always attained,
    i.e. $Z_i(\bpr)\not=\emptyset$, for all $\bpr\in\R^n$.
  \item For $\bpr\in \R^n$, we have
\begin{equation}\label{equ:arg-conj}
   \begin{split}
 Z_i(\bpr)=\argmax_{\ba\in\R^n} \{
\nubi{ \ba\cdot \bB}-\ba\cdot \bpr \}.
   \end{split}
\end{equation}
  \item \label{ite:dich}
Either $Z_i(\bpr)=\set{\ba}$ for some $\ba\in\R^n$ or
    $Z_i\subseteq \eR^n\setminus \R^n$.
\item \label{ite:last}
$Z_i(\bpr)=\set{\ba}$ if and only if
$\EE^{\hqi{\ba}}[\bB]=\bpr$ (in particular, $\bpr\in\Pui$).
  \end{enumerate}
\end{proposition}
\begin{proof}\
\begin{enumerate}
\item It follows for the fact that the function $Z_i$ is upper
semi-continuous on the compact space $\eR^n$.
\item It suffices to observe that
(\ref{equ:nub-def}) implies that
\begin{equation}
   \label{equ:repr-sup}
   \begin{split}
 \bu_i(\bpr)=-\exp\{-\gamma_i\underset{\ba \in \bar{R}^{n}}{\sup
}(\nubi{\ba\cdot\bB}-\ba\cdot\bpr)\}\cdot\big(-\ui{0}\big),\text{ for all }\bpr\in\R^n.
   \end{split}
\end{equation}
\item The set $Z_i(\bpr)$ is convex, so if it contains a point in
  $\R^n$ and a point in $\eR^n\setminus \R^n$, it must contain
  infinitely many points in $\R^n$. This is in contradiction with the
  strict concavity of $ b_i$ on $\R^n$.
\item Proposition \ref{pro:first-der} states that $b_i$ is
continuously differentiable on $\R^n$ and that
$\nabla b_i(\ba)= \EE^{\hq{-\ba}_i}[\bB]$.
Therefore, $\nubi{\ba\cdot\bB}-\ba\cdot\bpr$ is a
concave and differentiable function of $\ba\in\R^n$ and
its derivative is given by
$\EE^{\hq{-\ba}_i}[\bB]-\bpr$. Consequently,
$\nubi{\ba\cdot\bB}-\ba\cdot\bpr$ attains its maximum on $\R^n$ if
and only if $\EE^{\hq{-\ba}_i}[\bB]=\bpr$ has a solution
$\ba\in\R^n$. In that case, $Z_i(\bpr)=\set{\ba}$.
\end{enumerate}
\end{proof}

\begin{proposition}
\label{pro:equ-char}
  A pair $(\hat{\bpr},\hat{\ba})$
is a PEPQ if and only if $\hat{\bpr}\in \Pu_1\cap \Pu_2$, $\ba\in\R^n$ and
\begin{equation}\label{equ:pepq}
\EE^{\hqone{\hat{\ba}}}[\bB]=\EE^{\hqtwo{-\hat{\ba}}}[\bB]=\hat{\bpr}.
\end{equation}
\end{proposition}
\begin{proof}
If $(\hat{\bpr},\hat{\ba})$ is a PEPQ, then
$Z_i(\hat{\bpr})\cap\R^n\not=\emptyset$ and, so, by
Proposition \ref{pro:demand}, part (\ref{ite:dich}),
 we must have
$Z_i(\hat{\bpr})=\set{\ba_i}$, for some $\ba_i\in\R^n$ and
$\hat{\bpr}\in\Pui$, for $i=1,2$.
By (\ref{equ:pepq-condition}), we have $\ba_1=-\ba_2$.
 The equalities in (\ref{equ:pepq}), with $\hat{\ba}=\ba_1$
follow directly from part (\ref{ite:last}) of Proposition
\ref{pro:demand}.

Conversely, suppose that (\ref{equ:pepq}) holds. Then, by part
(\ref{ite:last}) of Proposition \ref{pro:demand}, we have
$Z_1(\bpr)=\set{\hat{\ba}}$ and
$Z_2(\bpr)=\set{-\hat{\ba}}$, which, in turn, implies
 (\ref{equ:pepq-condition}).
\end{proof}
We have also shown the following result which will be used shortly:
\begin{corollary}\label{cor:pepq}
  A pair $(\hat{\bpr},\hat{\ba})\in (\Pu_1\cap \Pu_2)\times \R^n$
is a PEPQ if and only if
\[ w_1(\hat{\ba})-b_2(\hat{\ba}) \leq w_1(\ba)-b_2(\ba)\text{ for any
}\ba\in\R^n,\text{ and $\hat{\bpr}=\nabla w_1(\hat{\ba})$.}\]
\end{corollary}
The main result of this Section is presented in the following Theorem:
\begin{theorem}
\label{thm:main-equ} Let $\EN_1,\EN_2\in\linf$,
$\gamma_1,\gamma_2>0$ and $\bB\in (\linf)^{n}$ be arbitrary, and
suppose that the Assumption \ref{ass:no-repl} is satisfied. Then,
there exists a unique partial equilibrium price-quantity
  $(\ba,\bpr)\in\R^n\times \R^n$. Moreover, $\bpr\in\Pu_1\cap \Pu_2$.
\end{theorem}
\begin{proof}
If the PEPQ $(\ba,\bpr)$ exists, then $\ba$ globally minimizes the
strictly concave function $w_1-b_2$, so it must be unique. To
establish existence, it will be enough to solve the equation
$\nabla f=0$, where $f=w_1-b_2$. Assume, to the contrary, that
$\nabla f(\ba) \not =0$, for all $\ba\in\R^n$. Continuity of $f$
implies that for each $m\in\N$ there exists $\ba_m\in
\overline{B}_m=\sets{ \ba\in\R^n}{ \sum_{i=1}^n
  \abs{\alpha_i} \leq m}$ such that $f(\ba_m) \leq f(\ba)$ for all
$\ba\in \overline{B}_m$. Thanks to strict convexity of $f$ and the
fact that $\nabla f\not = 0 $ on $\overline{B}_m$, we must have
$\norm{\ba_m}_1=m$, where $\norm{\ba}_1=\sum_{i=1}^n
\abs{\alpha_i}$.
In order to reach a contradiction, it will be enough to show that
\begin{equation}
   \label{equ:limit}
   \begin{split}
\liminf_{m\to\infty} \tfrac{f(\ba_m)}{m}>0.
   \end{split}
\end{equation}
Indeed, (\ref{equ:limit}) would provide the following coercivity
condition
\[\liminf_{m\to\infty} \inf
\sets{ \frac{f(\ba)}{\norm{\ba}_1}}{\ba\in
  \overline{B}_m\setminus\set{\bz}}>0,\] which, in turn, would guarantee
existence of a global minimizer $\ba_0\in\R^n$ for $f$ (see
Chapter 1 of \citet{BorLew00}), at which $\nabla f(\ba_0)=0$ holds.

The first step in the proof of (\ref{equ:limit}) uses the
representation (\ref{equ:con-to-uncon}) and the risk-measure
properties of $\nuwg{\cdot}$ to obtain the following:
\begin{equation}
   \label{equ:something-one}
   \begin{split}
 \liminf_{m\to\infty} \frac{f(\alpha_m)}{m} & =
\liminf_{m\to\infty} \frac{1}{m} \Big( \unuw{\gamma_1}{\ba_m\cdot
  \bB-\EN_1}+
\unuw{\gamma_2}{-\ba_m\cdot \bB-\EN_2}
\Big) \\
&\geq
\liminf_{m\to\infty} \frac{1}{m} \Big( \unuw{\gamma_1}{\ba_m\cdot
  \bB}-\norm{\EN_1}_{\linf}+
\unuw{\gamma_2}{-\ba_m\cdot \bB}-\norm{\EN_2}_{\linf}
\Big)\\
&=
\liminf_{m\to\infty} \Big(
\unuw{m\gamma_1}{\tfrac{1}{m}\ba_m\cdot  \bB}
+\unuw{m\gamma_2}{-\tfrac{1}{m}\ba_m\cdot  \bB}
\Big)\\
   \end{split}
\end{equation}
Any subsequence of $\N$ through which the last limit
inferior in (\ref{equ:something-one}) above is realized admits a
further subsequence $\sq{k}{m_k}$ such that
the sequence $\tfrac{1}{m_k} \ba_{m_k}$ converges
to some $\ba_0\in\R^n$ with $\norm{\ba_0}_1=1$; indeed,
the full sequence $\sq{m}{\tfrac{1}{m} \ba_m}$
takes values in the compact set $\sets{\ba\in\R^n}{\norm{\ba}_1=1}$.
Proposition \ref{pro:joint-cont}
implies that
\begin{equation}
   \label{equ:lims}
   \begin{split}
 \unuw{m_k\gamma_1}{\tfrac{1}{m_k}\ba_{m_k}\cdot\bB} & \to
\sup_{\QQ\in\MM_e} \EE^{\QQ}[\ba_0\cdot \bB], \text{ and}\\
 \unuw{m_k\gamma_2}{-\tfrac{1}{m_k}\ba_{m_k}\cdot\bB} & \to
-\inf_{\QQ\in\MM_e} \EE^{\QQ}[\ba_0\cdot \bB],
   \end{split}
\end{equation}
as $k\to\infty$. Therefore,
\[ \liminf_{m\to\infty} \tfrac{1}{m} f(\ba_m)=\sup_{\QQ\in\MM_e}
\EE^{\QQ}[\ba_0\cdot \bB]-\inf_{\QQ\in\MM_e} \EE^{\QQ}[ \ba_0\cdot
\bB].\]
It remains to note that the equality $\sup_{\QQ\in\MM_e}
\EE^{\QQ}[\ba_0\cdot \bB]=\inf_{\QQ\in\MM_e} \EE^{\QQ}[ \ba_0\cdot
\bB]$ cannot hold; if it did, Assumption \ref{ass:no-repl} would be
violated.
\end{proof}
\begin{remark}\label{rem:pepq}\
  \begin{enumerate}
  \item When $n=1$, the proof above can be simplified considerably;
    one can show that
\[
\lim_{\alpha\to\infty} w_1'(\alpha) >
\lim_{\alpha\to\infty} b_2'(\alpha)\text{ and }
\lim_{\alpha\to-\infty} w_1'(\alpha) <
\lim_{\alpha\to-\infty} b_2'(\alpha),\]
and deduce the existence of the solution of the equation
$w_1'(\alpha)=b_2'(\alpha)$ directly.

In addition, by Remark \ref{rem:approx-interval}, we easily get that the quantity $\tilde{\alpha}=\frac{\EE^{\hq{-\gamma_2 \EN_2}}[B]-\EE^{\hq{-\gamma_1 \EN_1}}[B]}
{\gamma_1\Delta^{\hq{-\gamma_1\EN_1}}(B)+\gamma_2\Delta^{\hq{-\gamma_2\EN_2}}(B)}$ minimizes the second order approximation
of the difference $w_1(\alpha)-b_2(\alpha)$. In view of Corollary \ref{cor:pepq}, we can heuristically consider $\tilde{\alpha}$ as
an approximation of the partial equilibrium quantity (PEQ), provided that $\tilde{\alpha}$ is close to zero.
\item
   Corollary \ref{cor:agreement} and the discussion preceding it show
   that when
  $\frac{\gamma _{1}}{\gamma_2}\EN_{1}\sim\EN_{2}$,
  the unique PEPQ must be of the form $(0,\bpr)$,
  where $\bpr=\EE^{\hq{-\gamma_1\EN_1}}[\bB]=\EE^{\hq{-\gamma_2\EN_2}}[\bB]$
  for every $\bB$ which satisfies the Assumption \ref{ass:no-repl}.
In such cases $\bpr$ should not be interpreted as a price of $\bB$,
  since no
  transaction actually occurs. Furthermore, the strict agreement (in
  the sense of Definition \ref{def:agreement}) can then be reached for no
  contingent claim of the from $\ba\cdot\bB$, $\ba\in\R^n$.

Even when $\frac{\gamma
_{1}}{\gamma_2}\EN_{1}\nsim\EN_{2}$, there might exist claims for
which the PEPQ is of the form $(0,\bpr)$. In fact, PEPQ is of the form $(0,\bpr)$ if and only if
$\EE^{\QQ^{(-\gamma_{1}\mathcal{E}_{1})}}[\bB]=\EE^{\QQ^{(-\gamma_{2}\mathcal{E}_{2})}}[\bB]$ (see Proposition \ref{pro:approx-agree}).
As an example, consider a
claim $\bB$ which is independent of the the stochastic
process $\bS$, as well as the two random endowments.
The partial equilibrium price is then simply a certainty equivalent
$\bpr=\EE[\bB]=\EE^{\QQ^{(-\gamma
      _{1}\mathcal{E}_{1})}}[\bB]=\EE^{\QQ^{(-\gamma
      _{2}\mathcal{E}_{2})}}[\bB].$

If a vector of claims $\bB$ satisfies the Assumption
\ref{ass:no-repl} and its PEPQ is of the form $(0,\bpr)$, then
$\nuwone{\ba\cdot\bB} -\nubtwo{\ba\cdot\bB} >0$, for every $\ba
\in \mathbb{R}^{n}\setminus\{0\}$, i.e. $\ba \cdot \bB\notin\agset$ for
all $\ba\neq 0$. In other words, any trade in a nontrivial linear
combination $\ba\cdot \bB$ must make at least one of the agents
strictly worse off.
\end{enumerate}
\end{remark}

\appendix

\section{Conditional Indifference Prices}\label{sec:cond-price}

The subject of this Section is the conditional indifference price
and some of its properties. The results stated below are
not only very useful for our analysis mutually agreeability, they
may also be seen as interesting in their own right
 since they describe some of the aspects
of  indifference evaluation under the presence of random
endowment. Some new results about
the unconditional indifference price (see Lemma \ref{lem:B-1-2},
Propositions \ref{pro:alpha}, \ref{pro:monotone} and
\ref{pro:joint-cont}), as well as several generalizations of existing
results in the case of the conditional price (see Theorem
\ref{thm:dual-rep}, Propositions \ref{pro:lsc-h} and
\ref{pro:con-asym}) are exhibited.

\subsection{First properties}
We remind the reader that the \define{writer's and buyer's
conditional (relative) indifference prices} $\nuwg{B}$ and
$\nuwg{B}$, for $B\in\linf$, are defined as
\[
\nuwg{B}=\inf\sets{p\in\R}{p-B\in \ag},\
\nubg{B}=\sup\sets{p\in\R}{B-p\in \ag},
\] where $
\ag=\sets{B\in\linf}{ \ug{B} \geq \ug{0}}$, with the notation
introduced on page \pageref{sec:agree} at the beginning of
Section \ref{sec:agree}. Proposition \ref{pro:first-cond} collects
some basic properties of the indifference prices and its proof is
standard.
\begin{proposition}\
\label{pro:first-cond}
\begin{enumerate}
\item $\nubg{B}=-\nuwg{-B}$, for $B\in\linf$. \item When
$\EN\in\rinf$ (in particular, when $\EN$ is constant)
$\nuwg{\cdot}$ and $\nubg{\cdot}$ coincide with their
unconditional versions $\unuwg{\cdot}$ and $\unubg{\cdot}$. \item
More generally, we have
$\nuw{\gamma}{\cdot}{\EN}=\nuw{\gamma}{\cdot}{\EN'}$ and
$\nub{\gamma}{\cdot}{\EN}=\nub{\gamma}{\cdot}{\EN'}$ as soon as
$\EN\sim\EN'$.
\end{enumerate}
\end{proposition}
When $\EN$ is constant or, more generally, when $\EN\in\rinf$,
$\nuwg{\cdot}$ and
$\nubg{\cdot}$ are usually denoted by $\unuwg{\cdot}$ and
$\unubg{\cdot}$, and are called the (writer's and buyer's)
\define{unconditional indifference prices}.

\subsection{Conditional indifference prices as convex risk
measures}\label{subsec:Conditional} With the notation from
subsection \ref{sse:special-measures}, the conditional
indifference price $\nuwg{\cdot}$ can be understood as an
unconditional indifference price computed under the probability
measure $\PP_{-\gamma\EN}$. In particular, using the terminology
of \citet{FolSch04}, Section 4.8, the
following statement holds:
\begin{proposition}
\label{pro:risk-measures}
  Maps $B\mapsto \nuwg{-B}$ and $B\mapsto-\nubg{B}$
are replication-invariant convex risk measures  on $\linf$, where
replication-invariance refers to the following property
\[ \nuwg{B+(\bt\cdot\bS)_T}=\nuwg{B},\text{ for all } \bt\in\bT.\]
\end{proposition}
 Moreover, these measures admit a robust dual representation, as
 stated in the following theorem, which follows from Theorem 2.2 in \citet{DelGraRheSamSchStr02} and Theorem 2.1 in \citet{KabStr02}:
\begin{theorem}(Delbaen F., Grandits P., Rheinl\"{a}nder T., Samperi D., Schweitzer M. and Stricker C., 2002, Kabanov Y. and Stricker C.,
2002)\\
\label{thm:dual-rep} For $B\in\linf$, we have
\begin{equation}
\label{equ:dual-representation}
\nuwg{B} =\underset{\QQ\in \mathcal{M}_{a}}{\sup }%
\left\{ \mathbb{E}_{\QQ}\left( B\right) -\frac{1}{\gamma
}h_{-\gamma\EN}(\QQ)\right\},
\end{equation}
where, for $C\in\linf$, we define the map
$h_{C}:\lone\mapsto[0,+\infty]$ as
\[h_{C} (\QQ)\label{pen.funct.}=\begin{cases}
\mathcal{H}(\QQ|\PP_C)-\mathcal{H}(\QQ^{(C)}|\PP_C) & \text{when
}\QQ\in \mathcal{M}_{a},\\ +\infty & \text{otherwise.}
\end{cases}
\]
The supremum in (\ref{equ:dual-representation}) is uniquely
attained by the measure $\QQ^{(-\gamma\EN+\gamma B)}$, which
belongs in $\mathcal{M}_{e,f}$ and its Radon-Nikodym derivative
with respect to $\PP_{-\gamma\EN+\gamma B}$ can be written as
\begin{equation}
   \label{equ:RN-6}
   \begin{split}
\frac{d\QQ^{(-\gamma\EN+\gamma B)}}{d\PP_{-\gamma\EN+\gamma
B}}=ke^{(-\gamma\bt ^{(-\gamma\EN+\gamma B)}\cdot \bS)_{T}},
   \end{split}
\end{equation}
where $\mathbf{\bt}^{(-\gamma\EN+\gamma B)}\in\bT$ is the maximizer of
the control problem associated with the value function $\ug{-B}$.
\end{theorem}
\begin{corollary}
\label{cor:nu-sc}
The maps $B\mapsto \nuwg{B}$ and $B\mapsto \nubg{B}$ are,
respectively,
lower and upper semi-continuous with respect to
the weak-* topology $\sigma(\linf,\lone)$.
\end{corollary}
\begin{proof}
It suffices to note that (\ref{equ:dual-representation}) represents
$\nuwg{\cdot}$ as a supremum of $\sigma(\linf,\lone)$-continuous and
linear functionals on $\linf$.
\end{proof}
The function $h_{-\gamma\EN}(\cdot)$ in Theorem \ref{thm:dual-rep}
is sometimes called the
\define{penalty function} for the indifference price $\nuwg{\cdot}$,
and is clearly convex (strictly convex on its effective domain $\MM_e,f$). It is well known (see,
e.g., \citet{FolSch04}, Lemma 3.29) that the  conjugate
representation,
\[ \EE[ X \log X]=\sup_{Y\in\linf} \left( \EE[ YX ] - \log
  \EE[e^Y]\right),\]
where we use the convention that $x\log(x)=+\infty$, for $x<0$, is
valid for all $X\in\lone$. Using this representation and
the natural identification of finite
measures equivalent to $\PP$ with their Radon-Nikodym derivatives in $\lone$, we can
readily establish the following properties of the penalty function
$h$:
\begin{proposition}
\label{pro:lsc-h}
  For $C\in\linf$, $h_C:\lone\mapsto [0,+\infty]$ is convex (strictly on its effective domain) and
  $\sigma(\lone,\linf)$-lower semicontinuous.
\end{proposition}

An immediate corollary of Proposition \ref{pro:lsc-h} and the
Hahn-Banach Theorem in the separation form (see \citet{FolSch04}
for details on convex analysis and \citet{JouSchTou06}, Theorem
2.1) is the following result:
\begin{proposition}\label{pro:penalty}
The map $h_{-\gamma\EN}$ is the minimal penalty function for
$\nuwg{\cdot}$, i.e.
\[h_{-\gamma\EN}(\QQ)\leq \tilde{h}(\QQ),\text{ for all }
\QQ\in\MM_a,\]
 whenever
the function $\tilde{h}$ satisfies
\[ \nuwg{B} =\underset{\QQ\in \mathcal{M}_{a}}{\sup }%
\Big( \mathbb{E}_{\QQ}\left( B\right) -\frac{1}{\gamma
}\tilde{h}(\QQ)\Big),\text{ for all $B\in\linf$. }\]
Moreover, we have the following, dual, conjugate representation
\[ \frac{1}{\gamma} h_{-\gamma\EN}(\QQ)=\sup_{B\in\linf}
\left(\EE^{\QQ}[B]- \nuwg{B}\right),\ \forall\,
\QQ\in\lone(\Omega,\FF_{T},\PP_{-\gamma\EN}).
\]
\end{proposition}

\subsection{Some auxiliary results}
Using the linearity of the set $\bT$ of the admissible trading
strategies and the properties of the exponential utility, one can
deduce (see \citet{Bec01a}, Chapter 1) that the following scaling
property holds true:
\begin{equation}
   \label{equ:scaling-property}
   \begin{split}
 \alpha \unuw{\alpha \gamma}{B}=
 \unuw{\gamma}{\alpha B},
\text{ for $B\in\linf$, $\gamma,\alpha >0$}.
   \end{split}
\end{equation}
The following Lemma (which is used several times in the present
paper) states that the risk measures induced by the indifference
price has a certain subadditive property, with true  additivity
holding only in exceptional cases.
\begin{lemma}
\label{lem:B-1-2}
For $B_1,B_2\in \mathbb{L}^{\infty }$ and $\gamma_1,\gamma_2>0$, let
$\tgamma>0$ be given by
$\frac{1}{\tgamma}=\frac{1}{\gamma_1}+\frac{1}{\gamma_2}$. Then,
\begin{itemize}
\item[(a)]
$\unuwone{B_1} +\unuwtwo{B_2} \geq  \unuw{\tgamma}{B_1+B_2}$,
  and
\item[(b)]
the following two conditions are equivalent
\begin{enumerate}
\item  $\unuwone{B_1} +\unuwtwo{B_2} =  \unuw{\tgamma}{B_1+B_2}$,
\item  $\frac{\gamma _{1}}{\tgamma}B_1\sim
\frac{\gamma_{2}}{
\tgamma }B_2$.
\end{enumerate}
\end{itemize}
\end{lemma}
\begin{proof}\
\begin{enumerate}
\item[(a)] Using the dual representation (\ref{equ:dual-representation}),
 the inequality in (a) above  is equivalent to the following inequality
\begin{multline}
   \label{equ:sum-sups}
\sup_{\QQ\in\MM_a} \Big(\EE^{\QQ}[B_1]-\frac{1}{\gamma_1} h(\QQ) \Big)+
\sup_{\QQ\in\MM_a} \Big(\EE^{\QQ}[B_2]-\frac{1}{\gamma_1} h(\QQ) \Big)
\geq
\sup_{\QQ\in\MM_a} \Big(\EE^{\QQ}[B_1+B_2]-\frac{1}{\tgamma} h(\QQ) \Big),
\end{multline}
which always holds for elementary reasons. \item[(b)] $(1) \Rightarrow
(2)$. If the equality in (1) above holds, then it also holds in
(\ref{equ:sum-sups}). By strict convexity of the function
$h(\cdot)$ in this effective domain, i.e. on $\MM_{e,f}$, and the
scaling property (\ref{equ:scaling-property}), this is equivalent
to  equality of dual minimizers
\[
\hq{\frac{\gamma_1}{\tgamma}B_1}=\hq{\frac{\gamma_2}{\tgamma}B_2}=\hq{
B_1+B_2}.
\]
By the representation (\ref{equ:RN-6}) of the Radon-Nikodym
derivatives of the above measures, we get
\begin{equation}
   \label{equ:}
   \nonumber
   \begin{split}
k_1 e^{(\bt ^{(\frac{\gamma_1}{\tgamma}B_1)}\cdot \bS)_{T}}
e^{\frac{\gamma_1}{\tgamma} B_1}&=
 \frac{d\QQ^{(\frac{\gamma_1}{\tgamma}B_1)}}
{d\PP_{\frac{\gamma_1}{\tgamma} B_1}}
\frac{d\PP_{\frac{\gamma_1}{\tgamma}}B_1}{d\PP}
=\frac{d\QQ^{(\frac{\gamma_1}{\tgamma}B_1)}}{d\PP}=\\
&=\frac{d\QQ^{(\frac{\gamma_2}{\tgamma}B_2)}}{d\PP}
=
 \frac{d\QQ^{(\frac{\gamma_2}{\tgamma}B_2)}}
{d\PP_{\frac{\gamma_2}{\tgamma} B_2}}
\frac{d\PP_{\frac{\gamma_2}{\tgamma}}}{d\PP}= k_2 e^{(\bt
^{(\frac{\gamma_2}{\tgamma} B_2)}\cdot \bS)_{T}}
e^{\frac{\gamma_2}{\tgamma} B_2},
   \end{split}
\end{equation}
and so $\frac{\gamma_1}{\tgamma} B_1-\frac{\gamma_2}{\tgamma} B_2=
(\bt\cdot \bS)_T+k$, where $k=\log(k_2)-\log(k_1)$ and $\bt=\bt
^{(\frac{\gamma_2}{\tgamma} B_2)}-\bt ^{(\frac{\gamma_1}{\tgamma}
  B_1)}$.

$(2)\Rightarrow (1)$. Conversely, suppose that
$\frac{\gamma_1}{\tgamma} B_1-\frac{\gamma_2}{\tgamma} B_2=
(\bt\cdot \bS)_T+k$, for some $k\in\R$ and $\bt\in\bT$.
Using the scaling property (\ref{equ:scaling-property}), the equality
in (1) is equivalent to
\begin{equation}
   \label{equ:equivalent-nus}
   \begin{split}
\frac{1}{\gamma_1} \unuwb{ \frac{\gamma_1}{\bar{\gamma}} B_1 }
+\frac{1}{\gamma_2} \unuwb{ \frac{\gamma_2}{\bar{\gamma}} B_2 }
=
\frac{1}{\bar{\gamma}} \unuwb{  B_1+B_2 }
   \end{split}
\end{equation}
By the risk equivalence between $\frac{\gamma _{1}}{\tgamma}B_1$
and $\frac{\gamma_{2}}{ \tgamma }B_2$ and the replication
invariance of $\unuwb{\cdot}$, we have
\begin{equation}
   \label{equ:equivalent-nus-2}
   \begin{split}
     \frac{1}{\gamma_1} \unuwb{ \frac{\gamma_1}{\bar{\gamma}} B_1 }
     +\frac{1}{\gamma_2} \unuwb{ \frac{\gamma_2}{\bar{\gamma}} B_2 } &=
     \frac{1}{\gamma_1} \unuwb{ \frac{\gamma_1}{\bar{\gamma}} B_1 }
     +\frac{1}{\gamma_2} \unuwb{ \frac{\gamma_1}{\bar{\gamma}}
       B_1+k+(\bt\cdot\bS)_T } \\
     &=\frac{1}{\bar{\gamma}} \unuwb{ \frac{\gamma_1}{\bar{\gamma}}
       B_1}+\frac{k}{\gamma_2}.
   \end{split}
\end{equation}
On the other hand,
\begin{equation}
   \label{equ:equivalent-nus-3}
   \begin{split}
     \frac{1}{\bar{\gamma}} \unuwb{ B_1+B_2 } &=
     \frac{1}{\bar{\gamma}}
     \unuwb{B_1+\frac{\gamma_1}{\gamma_2}B_1+\frac{\bar{\gamma}}{\gamma_2}
       (k+(\bt\cdot \bS)_T) }= \frac{1}{\bar{\gamma}}\unuwb{
       \frac{\gamma_1}{\bar{\gamma}} B_1 }+ \frac{k}{\gamma_2}.
   \end{split}
\end{equation}
The equality in (\ref{equ:equivalent-nus}) now follows directly from
(\ref{equ:equivalent-nus-2}) and (\ref{equ:equivalent-nus-3}).
\ \\[-8ex]
\end{enumerate}
\end{proof}

The conjugacy between (affine transformations of) $\nuwg{\cdot}$
and $h(\cdot)$, as displayed in Theorem \ref{thm:dual-rep} and
Proposition \ref{pro:penalty}, yields directly the following
auxiliary result:
\begin{lemma}
\label{lem:aux-nuw}
For $\EN,\tilde{\EN}\in\linf$, $\gamma>0$,
 the following two statements are equivalent
\begin{enumerate}
\item $\nuw{\gamma}{B}{\EN}\geq \nuw{\gamma}{B}{\tilde{\EN}}$,
for all $B\in\linf$,
\item $h_{-\gamma\EN}(\QQ)\leq h_{-\gamma\tilde{\EN}}(\QQ)$,
for all $\QQ\in\MM_{a}$.
\end{enumerate}
\end{lemma}
\noindent We use Lemma \ref{lem:aux-nuw} in the proof of the following
proposition:
\begin{proposition}\label{nonagree}
For $\EN\in\linf$ and $\gamma>0$, the following statements are
equivalent:
\end{proposition}
\begin{enumerate}
\item $\unuwg{B}\geq \nuwg{B}$, for all $B\in\linf$,
\item $\unuwg{B}= \nuwg{B}$, for all $B\in\linf$,
\item $\EN\in\rinf$, and
\item $\hq{0}=\hq{-\gamma\EN}$.
\end{enumerate}
\begin{proof}
$(4)\Rightarrow (3)$\ Just like in the proof of implication
$(1)\Rightarrow (2)$ in Lemma \ref{lem:B-1-2}, we can use the
equation (\ref{equ:RN-6}) in Theorem \ref{thm:dual-rep} to show
that (4) implies (3).

$(3)\Rightarrow (2)$\ Follows immediately from statement (3) in
Proposition \ref{pro:first-cond}.

$(2)\Rightarrow (1)$\ Clearly, (1) is  weaker than (2).

$(1)\Rightarrow (4)$ By Lemma \ref{lem:aux-nuw}, the equality in (2)
implies that
$h_{-\gamma\EN}(\QQ)\geq  h(\QQ)$, for all $\QQ\in\MM_a$, i.e.
\[
\HH(\QQ|\PP_{-\gamma\EN})-\HH(\hq{-\gamma\EN}|\PP_{-\gamma_\EN})
\geq \HH(\QQ|\PP)-\HH(\hq{0}|\PP),\ \forall\, \QQ\in\MM_a.
 \]
In particular, for $\QQ=\hq{-\gamma\EN}$, we get
\[
\HH(\hq{-\gamma\EN}|\PP)\leq \HH(\hq{0}|\PP).
\]
Therefore, $\hq{-\gamma\EN}=\hq{0}$, by the strict convexity of
the relative entropy $\HH(\cdot|\PP)$ on its effective domain.
\end{proof}

Considered as convex risk measure, the indifference price is not homogeneous.
In fact, the homogeneity holds only for replicable claims as the
following proposition states.
\begin{proposition}\label{pro:alpha}
For $B,\EN\in\linf$ and $\gamma>0$, the following statements are
equivalent:
\begin{enumerate}
\item $\nuwg{\alpha B}= \alpha \nuwg{B}$, for some $\alpha\in
  \R\setminus\set{0,1}$,
\item $B\in\rinf$.
\end{enumerate}
\end{proposition}
\begin{proof}
We assume, for simplicity, that $\EN=0$ (otherwise, we simply
change the underlying probability to $\PP_{-\gamma\EN}$).

$(2)\Rightarrow (1)$\ If $B\in\rinf$, then $\alpha B\in\rinf$, so (1)
follows from the replication-invariance of $\unuwg{\cdot}$.

$(1)\Rightarrow (2)$\ Suppose, first, that (1) holds with $\alpha>0$.
Then
\[ \sup_{\QQ\in\MM_a}\left( \EE^{\QQ}[B]-\frac{1}{\gamma}
h(\QQ) \right)= \sup_{\QQ\in\MM_a}\left(
\EE^{\QQ}[B]-\frac{1}{\alpha \gamma} h(\QQ) \right).\] The two
maximized functions are strictly concave, ordered and agree only
for $\QQ$ such that $h(\QQ)=0$. Therefore,  the equality of their
(attained) suprema forces the relation $h(\hq{\gamma
B})=h(\hq{\alpha\gamma B})=0$, which, in turn, implies that
$\hq{\gamma B}=\hq{\alpha\gamma B} =\hq{0}$. We can conclude that
$B\in\rinf$ by using the implication $(4)\Rightarrow (3)$ in
Proposition \ref{nonagree}.

It remains to treat the case $\alpha<0$. By considering the random
variable $\abs{\alpha} B$ instead of $B$, it is clear that we can
safely assume that $\alpha=-1$, i.e., $\unuwg{-B}=-\unuwg{B}$.
Equivalently, we have
 \[ \inf_{\QQ\in\MM_a}\left( \EE^{\QQ}[B]+\frac{1}{\gamma}
h(\QQ) \right)=
\sup_{\QQ\in\MM_a}\left( \EE^{\QQ}[B]-\frac{1}{ \gamma}
h(\QQ) \right),\]
which, by positivity of $h(\cdot)$, implies that $h(\hq{\gamma
  B})=0$. We continue as above to conclude that $B\in\rinf$.
\end{proof}
\subsection{Relationship between conditional and unconditional
  indifference prices}

It has been observed (see Remark 1.3.2 in \citet{Bec01a}) that the
conditional indifference price can be written as difference of two
unconditional ones:
\begin{equation}
\label{equ:con-to-uncon}
 \nuwg{B}=\unuwg{B-\EN}-\unuwg{-\EN}=\unuwg{B-\EN}+\unubg{\EN}.
\end{equation}
A similar relationship, namely $\unubg{B}=\unubg{B+\EN}-\unubg{\EN}=\unubg{B+\EN}+\unuwg{-\EN}$,
holds for the buyer's conditional indifference prices.

\subsection{$\nuwg{B}$  as a function of $\gamma $}

It is a known property of the {\em unconditional} indifference price
that the mappings $\gamma\mapsto \unuwg{B}$ and $\gamma\mapsto
-\unubg{B}$ are non-decreasing. In fact, we have the following, more
precise, statement
\begin{proposition}
\label{pro:monotone}
 For $\gamma>0$ and $B\in\linf$,
the mapping $\gamma\mapsto \unuwg{B}$ ($\gamma\mapsto \unubg{B}$) is
\begin{enumerate}
\item constant and equal to the value $\EE^{\QQ}[B]$,
  constant over $\QQ\in\MM_a$, when $B\in\rinf$, and
\item strictly increasing (decreasing), otherwise.
\end{enumerate}
\end{proposition}
\begin{proof} We only deal with the writer's price $\unuwg{B}$. The
  case of the buyer's price is parallel.
\begin{enumerate}
\item By the replication invariance of the $\unuwg{\cdot}$, the
value of $\unuwg{B}$ equals to the value $\EE^{\QQ}[B]$,
$\QQ\in\MM_a$, when $B\in\rinf$. \item Suppose now that
$\unuw{\gamma_1}{B}\leq \unuw{\gamma_2}{B}$, for some
$0<\gamma_1<\gamma_2$. By the dual representation
(\ref{equ:dual-representation}), we have
$\unuw{\gamma_1}{B}=\unuw{\gamma_2}{B}$, and using the scaling
property (\ref{equ:scaling-property}), we get
\[ \alpha \unuw{\gamma_2}{B}=
\unuw{\gamma_2}{ \alpha B},\]
where $\alpha=\gamma_2/\gamma_1>1$. By Proposition \ref{pro:alpha},
$B\in\rinf$.
\end{enumerate}
\end{proof}

A similar proposition in the conditional case fails. Indeed, here is a
simple example. Pick $\EN\not\in\rinf$, and set $B=\EN$, Then
$\nuwg{\EN}= \unubg{\EN}$ -  a {\em strictly decreasing} function of
$\gamma$. An even more instructive example in which the dependence of
$\gamma$ ceases to be monotone {\em at all} is given below.

\begin{example}
  We adopt the setting of Example \ref{exa:first}, and assume that the
  coefficients $b$ and $a$ are chosen in such a way that the distribution of the
  random variable $Y_T$ is diffuse (under $\PP$, and, therefore, under every
  equivalent martingale measure). Let $\QZ$ be the minimal-entropy
  martingale measure and let $g_i:\R\to\R$,
  $i=1,2$ be two bounded
  Borel-measurable functions.
  We set $\EN=-g_1(Y_T)$ and $B=g_2(Y_T)-g_1(Y_T)$, and compute the
conditional indifference price $\nuwg{B}$ as a difference
$\nuwg{B}=\unuwg{B-\EN}-\unuwg{-\EN}$. By the expression
(\ref{equ:con-to-uncon}) and the formula (\ref{ex1.formula}), we
have
\begin{equation}
   \label{equ:example-non-monotonw}
   \begin{split}
     \nuwg{B}&=\unuwg{g_2(Y_T)}-\unuwg{g_1(Y_T)} \\
&=
\frac{1}{\gamma (1-\rho^2)}\Big(
\ln \EE^{\QZ}[ \exp(\gamma (1-\rho^2) g_1(Y_T))] -
\ln \EE^{\QZ}[ \exp(\gamma (1-\rho^2)
g_2(Y_T))]\Big) .
   \end{split}
\end{equation}

The intervals of monotonicity of the mapping $\gamma\mapsto \nuwg{B}$
therefore coincide with the intervals of monotonicity of the function
$f:(0,\infty)\to\R$ given by
\[  f(\gamma)=\frac{1}{\gamma} \left(
\ln \EE^{\QZ}[ X_1^{\gamma}]- \ln \EE^{\QZ}[ X_2^{\gamma}]\right),\]
where the bounded and positive random variables $X_i$, are given by
$X_i=\exp((1-\rho^2) g_i(Y_T))$, $i=1,2$. It is clear that, thanks to
the assumption of diffusivity of the random variable $Y_T$, any
pair of probability distributions with compact support in  $(0,\infty)$
can be chosen for $X_1$ and $X_2$ by the appropriate choice of the
functions $g_1$ and $g_2$.

Thanks to the boundedness of $X_1$ and $X_2$, we can easily obtain the
following asymptotic expansion for the function $f$ around $\gamma=0$:
\[ f(\gamma)= \EE^{\QZ}[X_1]-\EE^{\QZ}[X_2]+\tot \gamma
(\Var_{\QZ}[X_1]-\Var_{\QZ}[X_2])+o(\gamma).\] In a similar
manner, we have
\[ \lim_{\gamma\to\infty} f(\gamma)= \ln
\norm{X_1}_{\linf}-\ln\norm{X_2}_{\linf}.\]
Therefore,  if $X_1$ and $X_2$  satisfy
\begin{enumerate}
\item $\EE^{\QZ}[X_1]<\EE^{\QZ}[X_2]$, and
\item $\Var_{\QZ}[X_1]<\Var_{\QZ}[X_2]$,
\end{enumerate}
the function $f$ is strictly decreasing and negative in a neighborhood of
$\gamma=0$. If, in addition
\begin{enumerate}
\setcounter{enumi}{2}
\item $\norm{X_1}_{\linf}>\norm{X_2}_{\linf}$,
\end{enumerate}
 holds, this trend   cannot continue
 for all $\gamma$
since
$f(+\infty)=\ln
\left(\norm{X_1}_{\linf}/\norm{X_2}_{\linf}\right)>0
>\EE[X_1]-\EE[X_2]=f(0+).$ The straightforward
construction of examples of
the random variables $X_1$
and $X_2$
having the above properties is left to the reader.
\end{example}

\subsection{Asymptotics of the conditional indifference prices}
The asymptotics of the unconditional indifference prices in the
risk-aversion parameter $\gamma$  are well-known (see, for instance,
Corollary 5.1 in \citet{DelGraRheSamSchStr02} or Proposition 1.3.4 in
\citet{Bec01a}):
\begin{equation}
   \label{equ:uncon-asym}
   \begin{split}
 \lim_{\gamma\to 0} \unuwg{B} & = \EE^{\QZ}[B],\ \lim_{\gamma\to +\infty}
 \unuwg{B}=\sup_{\QQ\in\MM_{e,f}} \EE^{\QQ}[B],\\
 \lim_{\gamma\to 0} \unubg{B} & = \EE^{\QZ}[B],\ \lim_{\gamma\to +\infty}
 \unubg{B}=\inf_{\QQ\in\MM_{e,f}} \EE^{\QQ}[B].
   \end{split}
\end{equation}
Using the decomposition (\ref{equ:con-to-uncon}),
these are easily extended to the conditional case:
\begin{proposition}
\label{pro:con-asym}
For $B,\EN\in\linf$, we have
\begin{equation}
   \label{equ:con-asym}
   \begin{split}
 \lim_{\gamma\to 0} \nuwg{B} & = \EE^{\QZ}[B],\
 \lim_{\gamma\to +\infty} \nuwg{B}=\sup_{\QQ\in\MM_{e,f}} \EE^{\QQ}[B-\EN]+\inf_{\QQ\in\MM_{e,f}}
\EE^{\QQ}[\EN], \text{ and } \\
\lim_{\gamma\to 0} \nubg{B} & = \EE^{\QZ}[B],\
 \lim_{\gamma\to +\infty} \nubg{B}=\inf_{\QQ\in\MM_{e,f}} \EE^{\QQ}[B-\EN]+\sup_{\QQ\in\MM_{e,f}}
\EE^{\QQ}[\EN], \text{ and }
   \end{split}
\end{equation}
\end{proposition}

One can, further, establish the continuous differentiability of the
map $\gamma\mapsto\nuwg{B}$, for $\gamma\in (0,\infty)$, by noting
the fact that \begin{equation}
   \label{equ:rep-gamma}
   \begin{split}
 \nuwg{B}= \frac{1}{\gamma }\left( \unuw{1}{\gamma (
    B-\EN) } - \unuw{1}{-\gamma \EN}\right)
   \end{split}
\end{equation}
 and using it together with
the
result of Theorem 5.3 in \citet{IlhJonSir05}
which states
that the function $\gamma\mapsto \unuw{1}{\gamma C}$ is continuously
differentiable on $(0,\infty)$ for $C\in\linf$.

For $n\in\N$, let $(\linf)^n$ denote the set of all $n$-tuples
$\bB=(B_1, \dots, B_n)$ of elements of $\linf$, with
$\norm{\bB}_{(\linf)^n}=\max_{k\leq n} \norm{B_k}_{\linf}$.  For
$\ba=(\alpha_1,\alpha_2,\dots, \alpha_n)\in\R^n$, we write
$\ba\cdot\bB=\sum_{k=1}^n \alpha_k B_k\in\linf$ and set
$\abs{\ba}=\max_{k\leq n} \abs{\alpha_k}$.
\begin{proposition}
\label{pro:joint-cont} For $\EN\in\linf$ and $\bB\in(\linf)^n$, the
function
$w:\R^n\times (0,\infty] \to \R$ given by
\begin{equation}
   \label{equ:def-bld}
   \begin{split}
 w(\ba,\gamma)=
\begin{cases}
  \nuwg{\ba\cdot\bB},& \gamma<\infty \\
 \sup_{\QQ\in\MM_{e,f}} \EE^{\QQ}[\ba\cdot\bB-\EN]+\inf_{\QQ\in\MM_{e,f}}
 \EE^{\QQ}[\EN],& \gamma=+\infty,
\end{cases}
   \end{split}
\end{equation}
is jointly continuous, and Lipschitz continuous on every subset $D$ of
its domain of the form
$D= [\gamma_0,\infty)\times\R^n$, $\gamma_0>0$.
\end{proposition}

\begin{proof}
The functional $B\mapsto \nuwg{B}$ is positive and coincides with
identity on constants, so for $\gamma\in (0,\infty)$,
\begin{equation}
   \label{equ:Lip-ld}
   \begin{split}
 \abs{\nuwg{\ba_1\cdot\bB}-\nuwg{\ba_2\cdot \bB}} \leq
\norm{(\ba_1-\ba_2)\cdot \bB}_{\linf} \leq
\abs{\ba_1-\ba_2}\norm{\bB}_{(\linf)^n}.
   \end{split}
\end{equation}
For $\gamma=+\infty$, the validity of (\ref{equ:Lip-ld}) follows
by passing to the limit $\gamma\to\infty$. On the other hand, by
(\ref{equ:rep-gamma}), for $B\in\linf$, $\gamma_0>0$ and
$\gamma_1,\gamma_2 \in [\gamma_0,\infty)$,  we have
\begin{equation}
   \label{equ:cont-gam}
   \begin{split}
 \abs{\nuw{\gamma_1}{B}{\EN}-\nuw{\gamma_2}{B}{\EN}}&\leq
\frac{1}{\gamma_0}
\left(
  \abs{\unuw{1}{\gamma_1 (B-\EN) }-\unuw{1}{\gamma_2 (B-\EN) }}
\right. \\
&\ \left. \qquad + \abs{\unuw{1}{-\gamma_1 \EN}-\unuw{1}{-\gamma_2\EN}}
\right)\\
&\leq \frac{1}{\gamma_0}(\norm{B-\EN}_{\linf}+\norm{\EN}_{\linf})
\abs{\gamma_1-\gamma_2}.
   \end{split}
\end{equation}
Therefore, for each $\gamma_0>0$,
there exists a constant $C=C(\gamma_0)>0$ such that
\[ \abs{w(\ba_1,\gamma_1)-w(\ba_2,\gamma_2)}\leq C \left(
  \abs{\gamma_1-\gamma_2} + \abs{\ba_1-\ba_2} \right),\text{ for }
\gamma_1,\gamma_2 \in [\gamma_0,\infty),\ \ba_1,\ba_2\in \R^n.\]
The existence of the limit $w(\ba,\infty)=\lim_{\gamma\to\infty}
w(\ba,\gamma)$ is the final ingredient in the proof.
\end{proof}

\section{The residual risk process}\label{sec:res-risk}
In this Section, we deal with the  the notion of the residual risk in
the dynamics setting, i.e., the residual-risk process.
We first recall the definition of dynamic version of
the
(conditional) indifference price.
\subsection{A dynamic version of the indifference price}
In addition to the study of the indifference prices $\nuwg{B}$ and
$\nubg{B}$ defined at time $t=0$, one can restrict attention to
any subinterval $[t,T]$ of $[0,T]$, and consider the filtered
probability space $(\Omega,\FF, \set{\FF_u}_{u\in [t,T]},\PP)$ and
  the stock-price process $\set{\mathbf{S}_u}_{u\in [t,T]}$. The (conditional)
indifference
  price of the contingent claim $B$, defined on this restricted model,
  is denoted by $\tnuwg{B}$. More precisely (see \citet{ManSch05}, Proposition 12, page 2127 for details),
  $\tnuwg{B}$ can be defined to
  be the a.s.-unique solution of the following equation
\begin{multline}
   \label{equ:dynam-indif}
     \esssup_{\bvt\in\bT} \EE\Big[ -\exp\Big(-\gamma\big( \EN+\int_t^T
       \bvt_u \, d\bS_u +\tnuwg{B}-B \big) \Big)\Big|\FF_t\Big]\\
=
     \esssup_{\bvt\in\bT} \EE\Big[ -\exp\Big(-\gamma\big( \EN+\int_t^T
       \bvt_u \, d\bS_u \big) \Big)\Big|\FF_t\Big].
\end{multline}
One can show using standard dynamic-programming methods (see e.g.
\citet{ManSch05}) that,
  when seen as a stochastic process, $\prf{\tnuwg{B}}$ admits a c\'
  adl\' ag modification. The process $\prf{\tnuwg{B}}$, modified so as
  to become c\' adl\' ag, is called the \define{writer's indifference
    price process} for the claim $B$. A natural analogue
corresponding to the buyer's price can be introduced in a similar fashion.

\subsection{The residual risk process}
Having defined the dynamic version $\prf{\tnubg{B}}$ of the
indifference price process, one can render the notion of the residual
risk introduced in Section \ref{sec:agree}, dynamic, too. More
precisely, the writer's residual risk process $\prf{\trrwg{B}}$ is
defined by
\[ \trrwg{B}= \tnuwg{B}-\nuwg{B}-\int_0^t \bt^{(B|\EN)}_u\, d\bS_u.\]
(note that $\Trrw{\gamma}{B}{\EN}=\rrwg{B}$). We can define the
buyer's residual risk process by
$R_t^{(b)}(B;\gamma|\EN)=\trrwg{-B}$. It is straightforward that
\begin{equation}
   \label{equ:dec-R}
   \begin{split}
     \trrwg{B}=\trrwg{B-\EN}-\trrwg{-\EN},\ t\in [0,T].
   \end{split}
\end{equation}
 and that the process $\prf{\trrwg{B}}$ admits a c\' adl\'
ag modification. It has been shown in \citet{ManSch05} (see Theorem
13) that when $\mathbb{F}$ is left-continuous, the residual risk
process admits a representation in terms of a martingale
orthogonal to $\bS$. We state the straightforward extension of
this result to the conditional case below.
\begin{theorem}[Mania M. and Schweizer M. (2005)]
\label{thm:man-sch}
Suppose that the
  filtration $\mathbb{F}$ is continuous, and let the process
  $\prf{\trrwg{B}}$ be as above. Then there exists a process
  $\prf{\tllwg{B}}$ such that
  \begin{enumerate}
    \item $\prf{\tllwg{B}}$ is a $\hq{-\gamma\EN}$-martingale in the space
      $BMO(\hq{-\gamma\EN})$, and
    \item $\trrwg{B}=\tllwg{B}-\frac{\gamma}{2} \ab{ \llwg{B} }_t$.
  \end{enumerate}
When $\EN\sim 0$, the family $\{\ullwg{B}\}_{\gamma>0}$ admits a
limit $\ullw{0}{B}$, as $\gamma\searrow 0$, in $BMO(\QZ)$. The
process $\ullw{0}{B}$ can be identified as a term in the
Kunita-Watanabe decomposition
\begin{equation}
   \label{equ:kun-wat}
   \begin{split}
B_t=E_{\QZ}[B]+\int_0^t \hat{\bt}^{(B)}_u\, d\bS_u+\tullw{0}{B},\ t\in [0,T],
   \end{split}
\end{equation}
of the $\QZ$-martingale $B_t=\EE^{\QZ}[B|\FF_t]$, where $\hat{\bt}^{(B)}$ is
an $\bS$-integrable predictable process for which  $(\hat{\bt}^{(B)}\cdot \bS)$
a $\QZ$-square integrable martingale. In particular,
$\ullw{0}{B}$ is strongly orthogonal to any $\QZ$-local martingale of the
form $(\bt\cdot \bS)$, $\bt\in L(\bS)$.
\end{theorem}

\def\cprime{$'$} \def\cprime{$'$}

\end{document}